\documentclass[letterpaper,12pt]{amsart}

\usepackage{enumerate}

\usepackage{color}
\usepackage{ulem,ifthen,xcolor,xkeyval,pdfcolmk}

\usepackage[abbrev]{amsrefs}
\usepackage{amssymb,amsmath,amsthm}

\usepackage{amsfonts}
\usepackage{graphicx}

\usepackage{hyperref}

\newtheorem{theorem}{Theorem}[section]
\newtheorem{acknowledgement}[theorem]{Acknowledgement}

\newtheorem{corollary}[theorem]{Corollary}

\newtheorem{definition}[theorem]{Definition}

\newtheorem{lemma}[theorem]{Lemma}

\newtheorem{notation}[theorem]{Notation}

\newtheorem{proposition}[theorem]{Proposition}
\newtheorem{remark}[theorem]{Remark}

\newtheorem{Remark}[theorem]{Remark}

\def \L{\Lambda}

\def \<{\langle}
\def \>{\rangle}

\def \R{\mathbb R}

\def \D{{\mathcal D}}
\def \H{{\cal H}}

\def \H^0{{\cal H}^0 or}

\def \w{\omega}

\def \kf{\frak k}
\def \p{\partial}
\def \beq{\begin{equation}}
\def \eeq{\end{equation}}

\def \V{{\mathcal V}}

\def \n{\nabla}
\def \eref{\eqref}

\def \lrc{\lrcorner}

\def \ex{{\bf x}}
\def \J{{\bf J}}
\def \Ae{{\bf A}}
\def \i{{\bf i}}
\def \j{{\bf j}}

\usepackage{tikz}

\numberwithin{equation}{section}

\begin{document}

\title{Neumann domination for the Yang-Mills heat equation}

\author{Nelia Charalambous}
\address{Department of Mathematics and Statistics, University of Cyprus, Nicosia, 1678, Cyprus} \email[Nelia Charalambous]{nelia@ucy.ac.cy}

\author{Leonard Gross}
\address{Department of Mathematics, Cornell University,  Ithaca, NY 14853-4201, USA}
\email[Leonard Gross]{gross@math.cornell.edu}

\thanks{N. Charalambous  was supported in part by NSF Grant DMS-0072164,
 NSF Grant DMS-0223098, by CONACYT of Mexico and by Asociaci\'{o}n
 Mexicana de Cultura A.C..}
 \date{\today}

\subjclass[2010]{Primary; 35K58, 35K65, Secondary;  70S15, 35K51, 58J35.}

\keywords{Yang-Mills, heat equation,
manifolds with boundary, Gaffney-Friedrichs inequality, weakly parabolic,
 Neumann domination, long time behavior.}
\maketitle

\begin{abstract}
 Long time existence and uniqueness of solutions
to the Yang-Mills heat equation
have been
 proven over a compact 3-manifold with
  boundary
 for initial data  of finite energy.
In  the present paper we improve on previous estimates by using a
Neumann domination  technique that allows us
  to get much better pointwise bounds
on the magnetic field.
As in the earlier work, we focus on Dirichlet, Neumann and Marini boundary conditions.
        In addition, we show that the  Wilson Loop functions, gauge invariantly regularized,
 converge as the parabolic time goes to infinity.
\end{abstract}

\tableofcontents

\section{Introduction}

 A gauge invariant regularization method for  Wilson loop variables appears
to be  an unavoidable necessity for
  construction of quantized Yang-Mills fields.   The standard methods
  of regularizing a quantum field, that have been successful in studying scalar field theories,
  are inapplicable to gauge fields.
  Thus
  a simple weighted average $\int_{\R^3} f(x-y) A(y) d^3y$ destroys gauge invariance
  of the gauge potential $A$. Similar expressions, such as $\int_{R^n} f(x-y) F(y) d^n y$,
  with $F(y)$ the curvature of $A$, also destroy gauge invariance,
   both for a space average, with $n=3$, or a  Euclidean space-time average, with $n=4$.
  For a closed curve $C$ in $R^3$
    the Wilson loop variable, $W_C(A) \equiv trace\, T(\exp(\int_C A(x)\cdot dx))$,
 where $T$ denotes time ordering around the loop, is gauge invariant but highly singular
 as a function of $A$ when  $A$ varies over the very large space of typical gauge fields required
 in the quantized theory.
            The lattice regularization of
  these functions of the gauge fields  has been the only useful gauge invariant
  regularization procedure so far but has not  produced a continuum limit.

 Polyakov \cite{Pol1,Pol2} already observed that the vacuum expectation
 of continuum Wilson loop variables are likely to be zero for a non-commutative gauge group.
 They are zero in the electromagnetic case.
 Nevertheless
 it has been hoped that the informal symbol, defined as $W_C(A)/ \<W_C\>_{VEV}$,
 which nominally is identically infinite in absolute value, could play a central role in
 a gauge invariant formulation of some future  internally consistent quantized Yang-Mills theory.
Such a program was outlined by E. Seiler, \cite[Pages 163-181]{Sei}.
 Many steps toward carrying this out were made  by a renormalization group approach
  in a series of papers by T. Balaban.
 See e.g. \cite{Bal}.

       In \cite{CG1} we began a regularization program based on use of  the Yang-Mills heat
 equation for regularizing gauge fields. The magnetic energy of a classical gauge field
 $A$ over three dimensional space is $\int_{R^3}|B(x)|^2 d^3x$,  where $B$ is the magnetic field
 ($\equiv$ curvature) of $A$.  The Yang-Mills heat equation
 flows $A$ in the direction of the negative of the gradient of the magnetic energy.
  It  is a non-linear, weakly parabolic equation with difficulties of its own.
 But it is fully gauge invariant: if one transforms the initial data $A_0$ by a gauge transformation
  on $\R^3$ and then propagates, one arrives at  the same gauge field as if one first propagates
 $A_0$ and then gauge transforms. Moreover the flow regularizes the initial data
 well enough so that
 the Wilson loop function $W_C(A(s))$ is meaningful for any fixed time $s >0$,
even  when $W_C(A_0)$ itself  is meaningless. Most importantly,  $W_C(A(s))$
is gauge invariant under  gauge transforms of the initial data $A_0$.
Here $A(s)$ is the solution to the Yang-Mills heat flow equation at time $s$.
The Yang-Mills heat flow has also been used for regularization as part of  a
  method for  implementing a Monte Carlo computational protocol for lattice gauge theory,
  \cite{L1,L2,L3,LW} .

 Like other heat equations, the Yang-Mills heat equation propagates information
 instantly. This would cause problems for local quantum field theory because
 one    wants $W_C(A(s))$  to capture information about $A_0$ just in  a small neighborhood
 of the curve $C$, not over all of $\R^3$. This issue can be resolved by using
 the Yang-Mills heat  equation regularization  over a bounded open
  set $M$ in $\R^3$ that contains $C$.
For this procedure one must prove existence and uniqueness of the solution
  when the initial data  is specified only in  $M$.  Of course for uniqueness one needs then to specify
  boundary conditions on
   the solution $A(s)$ for $ s >0$.
  These in turn
  must be gauge invariant and must allow use of initial data which are the restrictions to $M$
  of a typical gauge field $A_0$ on $\R^3$.  The classical Neumann  and Dirichlet boundary
  conditions will be vital boundary conditions for us for technical use. But in the end
  Marini boundary conditions, which simply set the normal
  component of
  the magnetic field $B(s)$  to zero on the boundary of $M$,  are the only ones that
  are fully gauge invariant. We will explore all three boundary conditions in this paper.

      We used, in \cite{CG1},  the Zwanziger-Donaldson-Sadun
\cite{Z}, \cite{Do1},  \cite{Sa} method for proving existence
of solutions to the Yang-Mills heat equation, which consists of adding a gauge symmetry breaking
term to the equation and then removing it from the solution by gauge transformation.
 The ZDS procedure does not enter directly into the present paper since our goal
 is to establish further properties of a solution whose existence we already know.
 Instead we will use the fact that the absolute values $|B(s, x)|$ and $|(d/ds) A(s, x)|$ satisfy
 parabolic inequalities with Neumann-like  boundary conditions.
 Our goal is to get detailed information about the behavior of these two
  functions as $s\downarrow 0$
  in order to help pass, eventually, to more general initial data.
Some of the initial steps in this technique will be carried out
over a compact manifold with boundary
rather than just over a bounded open set in $\R^3$ because they provide illumination
 as to what the techniques depend on, and there is little extra cost.

           In \cite{CG1} we established  existence and uniqueness of solutions
   in case the initial data $A_0$ is in the
 Sobolev space  $H_1(M)$. This corresponds to initial data of finite magnetic energy.
  In order to get this   program to work we anticipate that it will be necessary
  to extend the results in \cite{CG1} so as to allow  the initial data to lie in the
   larger space $H_{1/2}(M)$, which corresponds to initial data of finite magnetic action.
In the present paper we will still focus on initial data in $H_1$. However this is already
broad enough to include gauge fields that need to be regularized before their
Wilson loop functional can be defined.
           We will give an example in Section \ref{secLTB} of a current distribution in $\R^3$ whose magnetic
 field has finite energy but nevertheless gives infinite magnetic flux through certain loops,
 rendering the Wilson loop functional  for these loops meaningless.

Although our main concern is the  behavior of the solution for small time,
   we are also going to prove  that the
   Wilson loop functions  $W_C(A(s))$ converge as $s \rightarrow \infty$ for any initial
   gauge potential $A_0$ in $H_1$.

\section{Neumann domination}    \label{secNdom}

\begin{notation} \label{not2.1}{\rm $M$ will denote a compact
Riemannian 3-manifold with smooth boundary.  We will be concerned
 with a product bundle $M\times \V \rightarrow M$, where $\V$ is a finite dimensional
 real or complex vector space with an inner product.
$K$ will denote
a compact connected subgroup of the orthogonal, respectively, unitary group of
the space $End\ \V$, of operators on $\V$ to $\V$.
The Lie algebra of $K$, denoted $\frak k$, may then be identified with a
 real  subspace of $End\ \V$.
 We denote by  $\<\cdot, \cdot\>$  an  $Ad\ K$
 invariant inner product on $\frak k$ and denote its associated norm
 by $|\xi |_{\frak k}$ for $\xi \in \frak k$. We will not distinguish between
  $|\xi|_{\frak k}$ and $|\xi|_{End \V}$, which are equivalent norms.

 If $\w$ and $\phi$ are $\frak k$ valued $p$-forms  define
 $(\w, \phi) = \int_M\<\w(x), \phi(x)\>_{\L^p\otimes \frak k} d\, \text{Vol}$
 and $\|\w \|_2^2 = (\w, \w)$.
 Define also
 $\|\w\|_\infty = \sup_{x \in M}|\w(x)|_{\L^p \otimes \frak k}$ and
 \beq
 \| \w \|_{W_1(M)}^2
 = \int_M |\n \w|_{\Lambda^p\otimes\frak k}^2 d\, \text{Vol}\ \ + \| \w \|_2^2   \label{ymh2}
 \eeq
 where $\n$ is the Riemannian gradient on forms and $\n \w$ refers to
 the weak derivative. Define
 $W_1 = W_1(M) = \{ \w: \|\w\|_{W_1(M)} <\infty\}$.
 Since we are concerned only with a product bundle, a connection form
  can be identified with a $\frak k$ valued 1-form.
  For a connection form $A$, given in local coordinates by
  $A = \sum_{j=1}^3 A_j(x) dx^j$, its curvature (magnetic field) is given by
   \beq
 B = dA +(1/2) [A\wedge A]                                \label{ymh3}
 \eeq
 where $[A\wedge A] = \sum_{i,j} [A_i, A_j] dx^i\wedge dx^j$ and
 $[A_i(x), A_j(x)]$
 is the commutator in $\frak k$. $B$ is a   $\frak k$ valued 2-form.
 For $\w \in W_1$
  we define  $d_A \w = d \w + (ad\ A) \wedge \w$ and
  $d_A^* \w = d^*\w + (ad\ A\wedge)^* \w$. Boundary conditions
   will be imposed on these operators later.
  }
   \end{notation}

 We recall from \cite{CG1} the definition of a strong solution of the
    Yang-Mills heat equation.

\begin{definition}\label{defstrsol} {\rm Let $0 < T \le \infty$.
By a {\it strong solution} to the Yang-Mills heat
equation over $[0, T)$
we mean a continuous function
\beq
A(\cdot): [0,T) \rightarrow W_1 \subset \frak k\text{-valued 1-forms}
\eeq
such that
\begin{align}
a)& \ B(t) \in W_1 \  \text{for each}\ \  t\in (0,T),\
                  \text{where}\ B(t) =  \text{ curvature of}\ A(t),                \label{ymh8}\\
b)& \  \text{the strong $L^2(M)$ derivative $A'(t) \equiv dA(t)/dt $}\
                                      \text{exists on}\ (0,T),   \nolinebreak         \label{ymh9}    \\
c)& \  A'(t) = - d_{A(t)}^* B(t)\ \ \text{for each}\ t \in(0, T).               \label{ymh10}
\end{align}
A strong solution will be called {\it locally bounded} if
\begin{align}
&d) \  \| B(t)\|_\infty\ \text{is bounded on each
                             bounded interval $ [a,b) \subset (0, T)$ and}\    \label{ymh11}\\
&e) \  t^{3/4} \| B(t)\|_\infty\
           \text{is bounded on some interval $(0, b)$ with $0<b <T$.} \label{ymh12}
\end{align}
}
\end{definition}

\bigskip
\noindent
We are interested in three kinds of boundary conditions.

\bigskip
\noindent
{\it Neumann boundary conditions:}
\begin{align}
&i)\ \ \  A(t)_{norm} =0\ \ \text{for}\ t \ge 0 \ \text{and}    \label{N1}\\
&ii)\ \ B(t)_{norm}=0\ \   \text{for}\  t >0.         \label{N2}
\end{align}
\noindent
{\it Dirichlet boundary conditions:}
\begin{align}
 &i)\ \ \  A(t)_{tan} =0\ \ \ \text{for}\ t \ge 0 \ \text{and}    \label{D1}\\
  &ii)\ \  B(t)_{tan} =0 \ \ \ \text{for}\ t >0.                 \label{D2}
 \end{align}
{\it Marini boundary conditions:}
 \beq
 B(t)_{norm}=0\  \ \text{for}\ \ t >0.         \label{M1}
   \eeq

Given a solution of the Yang-Mills  heat equation, \eref{ymh10},
we are going to make pointwise estimates of
$|B(s,x)|_{\Lambda^2\otimes \frak k}$ and
$|A'(s,x)|_{\Lambda^1\otimes \frak k}$
based on  parabolic inequalities that these functions satisfy
for the Neumann Laplacian on real valued functions over $M$.
The final step in our method will require that  $\p M$ be convex in the sense that the
second fundamental form be non-negative on $\p M$.


\subsection{Sub-Neumann boundary conditions}   \label{secNdom1}

In this section $M$ will denote a compact, $n$-dimensional, Riemannian
manifold with a smooth, not necessarily convex,  boundary.

        \begin{proposition} \label{propsubNeu}
                    $($Sub-Neumann boundary conditions.$)$
Denote the extended shape operator by $Q(x) \in End( \L (T^* ( \p M))$
 $($the extension by derivation of the adjoint of the usual shape operator.
See e.g. {\rm \cite[Notation 4.6]{CG1}$)$.}
Denote by $\nabla_{\bf n}$
the outward drawn normal derivative operator.
  Let $A$ be a continuous $\frak k$ valued 1-form on $M$ and let $\w$ be a   $\frak k$ valued $p$-form on $M$ of class $C^1$.

a$)$ If
\beq
\w_{norm} = 0 \ \text{and}\ (d_A \w)_{norm} =0         \label{neu50}
\eeq
then
\beq
\nabla_{\bf n} | \w|^2 = -2\< \{I_{\frak k} \otimes Q\} \w, \w\>\
                                                             \ \text{on}\ \ \p M.  \label{neu51}
\eeq

b$)$ If
\beq
\w_{tan} =0 \ \ \text{and}\ \ (d_A^* \w)_{tan} = 0               \label{neu52}
\eeq
then
\beq
\nabla_{\bf n} | \w|^2 = -2\< \{I_{\frak k} \otimes(*^{-1} Q*)\} \w, \w\>\ \
                \text{on}\ \ \p M.                                                  \label{neu53}
 \eeq
\end{proposition}
         \begin{proof} In a neighborhood $U$ of the boundary choose
an adapted coordinate system $(U, x^1,\dots, x^n)$. See e.g. \cite[Notation 4.2]{CG1}.
 We can write a $\frak k$ valued $p$-form in $U$ as
\beq
\w = \beta + \gamma \wedge dx^n   \label{neu54}
\eeq
with $\beta = \sum_{J} \beta_J dx^J$ and
$\gamma = \sum_{I} \gamma_I dx^I$.
The multi-indices will signify generically  $J =(j_1,\dots, j_p)$
 with $ j_1<\cdots < j_p < n$ and
$I = (i_1,  \dots, i_{p-1})$ with $ i_1 < \cdots < i_{p-1} < n$.
Since
$\< dx^j, dx^n\> = 0$ for all $j <n$ we have
$$
|\w(x)|^2 = |\beta(x)|^2 + | \gamma(x) \wedge dx^n|^2    \ \ \text{in}\ U.
$$

      Suppose first that $\w_{norm} =0$ on $U\cap \p M$.
 That is, $\gamma|(U\cap \p M) =0$. Then, writing $\p_j^A = \p/\p x^j +ad\ A_j$
 on $\frak k$ valued functions, and $\nabla_j^A = \nabla_j +ad\ A_j$ for the Riemann covariant derivative on $\frak k$ valued forms, we have,  at $U\cap \p M$,
\begin{align}
(1/2) \p_n |\w(x)|^2
&= \< \n_n^A  \beta(x), \beta(x) \>
           + \< \n_n^A( \gamma(x) \wedge dx^n), \gamma(x) \wedge dx^n \> \notag \\
&= \< \n_n^A \beta(x), \beta(x)\>,   \label{neu55}
\end{align}
because $\gamma(x) = 0$ on $U\cap \p M$.
Now in $U$,
\beq
\n_n^A \beta = \sum_J (\p_n^A \beta_J) dx^J
     + \sum_J \beta_J(\n_n  dx^J).    \notag
     \eeq
 But
  $\sum_J \beta_J(\n_n  dx^J)
        =-\sum_J \beta_JQ dx^J
 =   - (I_{\frak k} \otimes Q) \beta$  at $U\cap \p M$.
 So
     \begin{align}
   \n_n^A \beta    = \sum_J(\p_n^A \beta_J) dx^J
            - (I_{\frak k} \otimes Q) \beta\ \     \text{at}\ \ U\cap \p M.     \label{neu56}
   \end{align}
   Further,
   \begin{align}
    d_A(\gamma \wedge dx^n)
    & = (d_A \gamma)\wedge dx^n \notag \\
    &= \sum_{j=1}^{n-1}
                  \sum_I(\p_j^A \gamma) dx^j\wedge dx^I \wedge dx^n \notag \\
    &=0 \ \ \text{on}\ \ U\cap \p M ,                                                \label{neu56.1}
    \end{align}
    because $\gamma$, as well as any  tangential derivative of $\gamma$, is zero
    when $\w_{norm} =0$.
           Hence,  if both equations in \eref{neu50} hold, then, by \eref{neu54}
   and \eref{neu56.1}, we find
    \begin{align*}
    0 = (d_A\w)_{norm} &= (d_A \beta)_{norm}\\
    &= \sum_J (\p_n^A \beta_J) dx^n \wedge dx^J.
    \end{align*}
    Thus $ \p_n^A \beta_J =0$ on $U\cap \p M$ and consequently \eref{neu56}
     reduces to
    \beq
    \n_n^A \beta = - (I_{\frak k} \otimes Q) \beta
                \ \ \text{on}\ \ U\cap \p M.                        \label{neu57}
    \eeq
    Therefore \eref{neu55} yields
      $(1/2) \p_n | \w |^2 = - \<(I_{\frak k} \otimes Q) \beta , \beta\>
    = -\<(I_{\frak k} \otimes Q)\w, \w \>$,
    which is \eref{neu51} since $\n_{\bf n} = \p_n$ in this coordinate system.
          In the last step we have used $\beta = \w$ on $U\cap \p M$.

   To prove the assertion in case b) one need only observe that
       if \eref{neu52} holds for $\w$ then \eref{neu50} holds for $* \w$.
    Consequently
        \begin{align*}
        \p_n | \w|^2 & = \p_n |*\w|^2 \\
         &= -2 \< (I_{\frak k} \otimes Q) *\w, , * \w \>,
         \end{align*}
         which is \eref{neu53}.
         \qed
\end{proof}


     \begin{corollary}   \label{corsubneu} In addition to the hypotheses of Proposition
\ref{propsubNeu}, suppose
 that $M$ is convex in the sense that its
second fundamental form is everywhere non-negative on $\p M$.
If  \eref{neu50} or \eref{neu52} hold then
\beq
\nabla_{\bf n}  | \w |^2 \le 0.                                     \label{neu58}
\eeq
If, in addition, $\p M$ is totally geodesic then
\beq
\nabla_{\bf n}  |\w|^2 =0 .                                           \label{neu59}
\eeq
\end{corollary}
   \begin{proof}
If $M$ is convex then  $Q(x) \ge 0$ on $\p M$ as is its unitary
 transform $*^{-1} Q(x) *$.
The inequality \eref{neu58} now follows from \eref{neu51} and \eref{neu53}.
If $\p M$ is totally geodesic then $Q(x) =0$ on $\p M$ as is
its transform  $*^{-1} Q(x) *$. \eref{neu59} now follows in the same way.
\end{proof}

           \begin{Remark}
 {\rm In the context of a Riemannian $n$-manifold without further vector
  bundle structure, the first author found, in \cite{Cha4},  that the Neumann boundary condition \eref{neu59}   follows from either of the
   boundary conditions $(a)  \ \w_{norm} = 0$ and $(d \w)_{norm} =0$ or
 $(b)  \ \w_{tan} =0$ and $(d^* \w)_{tan} = 0$, in the presence of a slightly weaker
 condition on the boundary than used in this paper. Namely, it was shown that if
 $\w$ is an $(n-1)$-form then it suffices that the trace of the second fundamental form be
 zero. But for lower order forms the condition that the boundary be totally geodesic was needed. In the present paper the weakened hypothesis would be applicable, in case dim $M =3$,
 to the curvature $B$ but not to $A'$.

       In \cite{Cha4}, in addition to the Hodge Laplacian, the first author considered the
  Bochner Laplacian and showed that, for the relevant notion of Dirichlet boundary
  condition on the $k$-form $\w$, no conditions on the boundary are needed to conclude
  \eref{neu59}.
 }
 \end{Remark}


\subsection{Domination by the Neumann heat kernel}    \label{secNdom2}

In this section $M$ will denote the closure of a  bounded open subset of $\R^n$
with smooth boundary.
$\Delta$ will denote the Laplacian on real valued functions on $M$
with domain $W_2(M)$ and $\Delta_N$ will denote the Neumann version.
Some aspects of the techniques we are exploring  have been used
for manifolds without boundary in \cite{Do1} and \cite{Sa} for the Yang-Mills heat equation.


              \begin{lemma}\label{lemNdom1}
     Let $\psi$ be a real valued function in $W_2(M)$  whose outer
 normal derivative satisfies
     \beq
      \n_{\bf n} \psi \le 0 \ \ \ \text{a.e. on }\ \p M.                                  \label{neu60}
      \eeq
      Then
      \beq
 e^{t\Delta_N} \Delta \psi
    \le \Delta_N e^{t\Delta_N} \psi\ \   a.e. \ \ \text{for all}\ \ t>0.              \label{neu61}
      \eeq
   \end{lemma}
                 \begin{proof}
    If $ 0 \le f \in \D(\Delta_N)\cap C^\infty(M)$ then $\n_{\bf n} f =0$ on $\p M$.
        Hence
  \begin{align*}
  (\Delta \psi, f) & = \int_M (div\ grad\ \psi) f dx \\
  & = \int_{\p M} ( {\bf n}\cdot grad\ \psi) f - \int_M\< (grad\ \psi), (grad\ f) \> dx\\
  &\le -  \int_M\< (grad\ \psi), (grad\ f) \> dx,
  \end{align*}
  in view of \eref{neu60}.
  Since $f \in \D(\Delta_N)$ we may integrate by parts once more to find
  \beq
  (\Delta \psi, f) \le ( \psi, \Delta_N f).     \label{neu65}
  \eeq
    Now let $ 0 \le \phi \in C_c^\infty(M^{\text{int}})$. We may
    apply \eref{neu65} to $f \equiv e^{t\Delta_N} \phi$ because
  $e^{t\Delta_N}$ is positivity preserving.
  It follows that
    $$
( e^{t\Delta_N} \Delta \psi, \phi)  =     (\Delta  \psi, e^{t\Delta_N} \phi )
     \le (\psi, \Delta_N e^{t\Delta_N} \phi)   = (\psi, e^{t\Delta_N}\Delta_N \phi).
  $$
  Hence
  $
   ( e^{t\Delta_N} \Delta \psi, \phi) \le ( \Delta_N e^{t\Delta_N} \psi, \phi)
  $
   for all non-negative $\phi \in C_c^\infty(M^{\text{int}})$.
 This proves \eref{neu61}.
 \end{proof}

          \begin{proposition}\label{propDomination}
Suppose that $M$ is  convex in the sense that the
  second fundamental  form is non-negative on $\p M$.
      Let $T>0$.    Suppose that
$A(\cdot) : [0,T) \rightarrow C^1( M; \Lambda^1\otimes \frak k)$
is a time dependent, 1-form on $M$
which is continuous in the time variable.
 Let $\w(\cdot):[0,T) \rightarrow  C^2( M; \Lambda^p\otimes \frak k)$
 be a time dependent, $\frak k$ valued, $p$-form on $M$
 which is continuously differentiable in the time variable  and
  satisfies the equation
 \beq
\w'(s, x) = \sum_{j=1}^n (\nabla_j^{A(s)} )^2 \w(s,x) + h(s,x),   \label{5.51}
\eeq
where  $h \in C([0,T)\times M; \Lambda^p\otimes \frak k)$.
Assume also that $\w$ satisfies either the boundary conditions
\beq
\w(s)_{norm} =0, \ \ \text{and}\ \  (d_{A(s)} \w(s))_{norm} =0
             \ \ \text{for all}\ \ s\in [0,T)                                   \label{5.53}
\eeq
or
\beq
\w(s)_{tan} = 0 \  \ \text{and} \ \ (d_{A(s)}^* \w(s))_{tan} = 0
       \  \ \text{for all}\ \ s\in [0,T).                                         \label{5.54}
\eeq
Then, for all $(t,x) \in [0,T)\times M$, there holds
\beq
|\w(t, x)| \le \{e^{t\Delta_N} |\w(0)|\}(x)
      + \int_0^t \{e^{(t-s)\Delta_N}  |h(s)|\}(x) ds .          \label{5.55}
\eeq
Here the norm denotes $|\cdot |_{\Lambda^p\otimes \frak k}$.
\end{proposition}
        \begin{proof}
Given $\w$ as specified, let $\epsilon >0$ and  define
\beq
 \psi(s,x) = ( |\w(s,x)|^2 + \epsilon^2)^{1/2}.        \label{5.41}
 \eeq
      For fixed $s$ (and suppressing $s$) we assert that
 \beq
 \<\sum_{j=1}^n (\nabla_j^A)^2 \w(x), \w(x)\>
          \le  \psi(x) (\Delta \psi)(x)\  \text{for all}\ x \in M^{int}.   \label{5.43}
 \eeq
 The proof of this well known pointwise inequality  follows a standard
  pattern and does not depend on the boundary conditions.
  Thus for any real valued function $\psi \in C^2(M)$ one verifies easily
   the identity
   \beq
\psi(x) \Delta \psi (x)
     = (1/2) \Delta \psi^2(x) - |grad\, \psi(x)|^2,\   \notag
\eeq
   and then, with $\psi$ defined now by \eref{5.41},  one computes, for each $s$,  that
 \beq
(1/2) \Delta \psi^2(x) = (1/2)\Delta |\w(x)|^2
    = \< \sum_{j=1}^n (\nabla_j^A)^2 \w(x), \w(x) \>
           + \sum_{j=1}^n |\nabla_j^A \w(x)|^2 .    \notag
\eeq
But
 $
 |\p_j \psi(x)| = |(1/2)(\p_j |\w(x)|^2)|/\psi(x)
= \<\n_j^A \w(x), \w(x)\>/\psi(x) \le |\n_j^A \w(x)|.
$
Combining this with  the previous two equations
 yields \eref{5.43}.

Suppose now that $\w(s,x)$ satisfies the differential equation \eref{5.51}.
Take the pointwise inner product of \eref{5.51} with $\w(s, x)$ to find,
with the help of \eref{5.43},
\begin{align*}
\psi(s, x) \psi'(s,x) &= (1/2) (d/ds) \psi^2(s, x)\\
&  =  \< \w'(s,x), \w(s,x)\>\\
& = \<\sum_{j=1}^3 (\nabla_j^{A(s)} )^2 \w(s,x) + h(s,x), \w(s,x)\>\\
&\le \psi(s,x) \Delta \psi(s, x) + |h(s,x)| |\w(s,x)|.
\end{align*}
Divide by $\psi(s,x)$ to deduce
\beq
 \psi'(s,x) \le \Delta \psi(s,x) +|h(s,x)|  \   \text{for all}\ x \in M^{int}. \label{5.53a}
\eeq
In view of \eref{5.53} and \eref{5.54}, it follows from
 Corollary \ref{corsubneu} that
 $\nabla_{\bf n} \psi^2= \nabla_{\bf n}  |\w|^2 \le 0$ on $\p M$.
 And, since $\psi \ge \epsilon >0$ on $M$, it follows that
 \beq
 \nabla_{\bf n} \psi \le 0\ \ \text{on}\ \ \p M.                \label{5.53b}
 \eeq
Let $\phi(s,x) = \psi(s,x) - \epsilon$. Then \eref{5.53a}  and \eref{5.53b} hold
 also with $\psi$
replaced by $\phi$. Thus
$\phi'(s,x) - \Delta \phi(s,x) - |h(s, x)| \le 0$ for each $x \in M^{int}$ and moreover
 $\nabla_n \phi \le 0$ on $\p M$.
 For $0 \le s \le t <T$ define at each $x\in M$ (and suppressing $x$)
$$
u(s) = e^{(t-s)\Delta_N} \phi(s)
     + \int_s^t e^{(t-\sigma)\Delta_N} |h(\sigma)| d \sigma.
 $$
 Then, by virtue of  Lemma \ref{lemNdom1},
\begin{align*}
 (d/ds)u(s) &=  -\Delta_N e^{(t-s)\Delta_N}\phi(s)
                      +e^{(t-s)\Delta_N} \{ \phi'(s) - |h(s)|\}\\
  & \le  e^{(t-s)\Delta_N}\{ -\Delta \phi(s) +  \phi'(s) - |h(s)|\}\\
  & \le 0,
 \end{align*}
  wherein we have used once more the fact that $e^{t\Delta_N}$ is positivity preserving
  for $t\ge 0$.  Thus $u(t) \le u(0)$. That is,
 \beq
 \phi(t)
 \le  e^{t \Delta_N} \phi(0)
       +\int_0^t e^{(t-\sigma)\Delta_N} |h(\sigma)| d \sigma.       \label{5.66}
 \eeq
 Now observe that $ 0 \le \phi(t,x) \le |\w(t,x)|$ and
    $\lim_{\epsilon\downarrow 0} \phi(s,x) = |\w(s,x)|$ for all $s$ and $x$.
  Using the dominated convergence theorem on the first term on the right of \eref{5.66} (which is an integral of a heat kernel), we may now let $\epsilon\downarrow 0$ in \eref{5.66} to arrive at \eref{5.55}.\qed
\end{proof}

      \begin{Remark} \label{remHunSim} {\rm If, in Proposition \ref{propDomination},
 one assumes that $A$ is independent of time and that $h \equiv 0$ then the inequality
 \eref{5.55} asserts that
 \beq
 |e^{t\Delta^A} \w(0) | \le e^{t\Delta_N} |\w(0)|,                \label{5.68}
 \eeq
 where $\Delta^A$ is the gauge covariant Laplacian on $\frak k$ valued
$p$-forms ($p \ge 1$)
associated to either relative or absolute boundary conditions and
$\Delta_N$ is the Neumann Laplacian on real valued functions.
This is a diamagnetic inequality (see \cite[Section 1.3]{CFKS}) for a region with boundary.
It seems 	quite feasible to derive our results from such an inequality
by writing the time dependent propagator as a limit of short time propagators
for $A(t)$ with different $t$. This would entail some regularity on the $t$
dependence of $A(t, \cdot)$. We have not explored this approach.
For a recent paper extending and reviewing  diamagnetic inequalities
of the form \eref{5.68} when $\frak k$ is abelian  see \cite{HS}.
}
\end{Remark}

\subsection{Pointwise bounds on solutions} \label{secNdom3}

Henceforth $M$ will denote the closure of a bounded open set in $\R^3$
with smooth boundary. We will assume  $M$ to be convex in the sense
that its second fundamental form is everywhere non-negative. For the Neumann
heat operator $e^{t\Delta_N}$ over $M$, the constant
\beq
c_N = \sup_{0<t \le 1} t^{3/4} \|e^{t\Delta_N}\|_{2\rightarrow \infty}  \label{AAhk}
\eeq
is  finite, \cite[page 274]{Tay3}. As in \cite{CG1}, we will take
$c:= \sup\{ |\,[\xi, \eta]\,|_\kf: |\xi|_\kf \le 1, |\eta|_\kf \le 1\}$ as a measure of the non-commutativity of  $\kf$.

        \begin{theorem}\label{thmAA} There exist strictly positive
 constants $a$ and $\gamma$ such that
        for any number $\tau \in (0, 1/2]$ and any smooth solution $A(\cdot)$ to
  the Yang-Mills heat equation \eref{ymh10} over the interval $[0,\infty)$
   satisfying either   Neumann boundary conditions \eref{N1} and \eref{N2},
   or Marini boundary conditions \eref{M1},
 or Dirichlet boundary conditions \eref{D1} and \eref{D2},
  the inequality
 \beq
 (2\tau)^{1/4} c\|B_0\|_2  \le a\ \ \        \label{AA0}
 \eeq
 implies that
 \begin{align}
\|B(t)\|_\infty &\le 2c_N \|B_0\|_2    t^{-3/4},\ \ \text{for} \ \      0 <t  \le 2\tau,\label{AA1} \\
\|B(t)\|_\infty &\le 2c_N \|B_0\|_2 \tau^{-3/4}, \ \  \text{for} \ \   \tau \le t <\infty    \ \ \text{and} \label{AA2}\\
\|A'(t)\|_\infty & \le \gamma \|A'(0)\|_2 t^{-3/4}, \ \  \text{for} \ \   0 < t \le  2\tau .       \label{AA3}
\end{align}
In particular $A(\cdot)$ is a locally bounded strong solution.
Moreover
\begin{align}
&\tau^{5/4} \|A'(t)\|_\infty  \le \gamma \|B_0\|_2, \ \text{for} \ \ 2\tau \le t <\infty \ \ \text{and}    \label{AA4} \\
&\|A'(t)\|_\infty \to 0 \ \text{as}\ t \to \infty .                      \label{AA5}
\end{align}
\end{theorem}

The proof depends on the following lemma.

\begin{lemma}\label{lemAA1} $($Differential identities$)$ For a smooth solution to
\eref{ymh10} there hold
\beq
d B(t)/dt = \sum_{j=1}^3 (\nabla_j^{A(t)})^2 B(t) + B(t)\# B(t) \ \ \text{and} \label{AA11}
\eeq
\beq
(d/dt) A'(t) =\sum_{j=1}^3 (\nabla_j^{A(t)})^2 A'(t)
+ B(t)\# A'(t) - [A'(t)\lrc B(t)],                                 \label{AA12}
\eeq
where $\#$ denotes a pointwise product of forms arising from the Bochner-Weitzenboch formula.
\end{lemma}
\begin{proof}
Bianchi's identity and the Bochner-Weitzenbock formula   yield
 \begin{align*}
 B'(t)& = d_{A(t)} A'(t) \\
      & = -( d_{A(t)}  d_{A(t)}^* +  d_{A(t)}^*  d_{A(t)})B(t) \\
      & = \sum_{j=1}^3 (\nabla_j^{A(t)})^2 B(t) + B(t)\# B(t),
      \end{align*}
      which is \eref{AA11}.
Differentiating \eref{ymh10}
 with respect to $t$ gives
\begin{align*}
A''(t) &= - d_{A(t)}^*B'(t) - [ A'(t) \lrc B(t)] \\
  &= -d_{A(t)}^*d_{A(t)} A'(t) - [ A'(t) \lrc B(t)].
  \end{align*}
  Since $ d_{A(t)}^*A'(t) = -d_{A(t)}^* d_{A(t)}^* B(t) = 0$ we find
  \begin{align*}
  A''(t) &= -(d_{A(t)}^*d_{A(t)}  +d_{A(t)} d_{A(t)}^*)A'(t) - [A'(t) \lrc B(t)]\\
   & = \sum_{j=1}^3 (\nabla_j^{A(t)})^2 A'(t) + B(t)\# A'(t) - [ A'(t)\lrc B(t)],
   \end{align*}
   which is \eref{AA12}.
\end{proof}

\bigskip
\noindent
\begin{proof}[Proof of Theorem \ref{thmAA}]
Both of the equations \eref{AA11} and \eref{AA12} have the form
specified in \eref{5.51} with different choices of the form $\w$ and the
function $h$. We need to verify the boundary conditions \eref{5.53}
or \eref{5.54} in each case.

First choose $\w(s,x) = B(s,x)$ and $h(s, x) = B(s,x)\# B(s,x)$.

    If $A(\cdot)$ satisfies Marini boundary conditions,    then
  $\w_{norm} = B_{norm} = 0$ by \eref{N2}, while $d_{A(t)} B(t) = 0$ by the Bianchi
  identity. So \eref{5.53} holds if  $A(\cdot)$ satisfies Marini boundary conditions.
   Since Neumann boundary conditions
  are a special case of Marini boundary conditions, \eref{5.53} holds in that case also.
        If $A(\cdot)$ satisfies Dirichlet boundary conditions then $\w_{tan} = B_{tan} = 0$
  by \eref{D2},
  while $(d_{A(t)}^* \w(t))_{tan} = (d_{A(t)}^* B(t))_{tan} = -A'(t)_{tan} =0$
  by \eref{ymh10} and \eref{D1}. So \eref{5.54} holds for Dirichlet boundary conditions also.
        In either case we may therefore apply Proposition \ref{propDomination}
        to the choice $\w= B$.
The inequality \eref{5.55}  then gives the following pointwise inequality
\begin{align}
|B(t, x)| &\le \{e^{t\Delta_N} |B(0)|\}(x)
      + \int_0^t \{e^{(t-s)\Delta_N} |B(s) \# B(s)|\}(x) ds.   \label{AA30}
      \end{align}
By \eref{AAhk},
\begin{align}
\|B(t)\|_\infty &\le     \| e^{t\Delta_N} |B_0| \|_\infty
      + \int_0^t \|e^{(t-s)\Delta_N} c |B(s)|^2\|_\infty ds    \notag\\
      &\le c_N \{ t^{-3/4} \|B_0\|_2 + \int_0^t(t-s)^{-3/4} c\| \ |B(s)|^2 \|_2 ds\} \label{5.75}
\end{align}
for  $0 < t \le 1$.
Define
\beq
\beta(t) = \sup_{0\le s \le t} s^{3/4} \|B(s)\|_{L^\infty(M)}.
\eeq
This is finite for each number $t\in (0,\infty)$   because $A(\cdot)$ is smooth.
Then
$
\|B(s)\|_\infty  \le s^{-3/4} \beta(t)   \ \ \text{whenever}\ \ 0 < s \le t,
$
and therefore
$$
\||B(s)|^2\|_2 \le s^{-3/4}\beta(t) \|B(s)\|_2 \le  s^{-3/4}\beta(t) \|B_0\|_2.
$$
Using this to estimate the integrand in \eref{5.75} we find
\begin{align}
\int_0^t(t-s)^{-3/4} c\| \ |B(s)|^2 \|_2 ds
&\le \beta(t) c\| B_0\|_2 \int_0^t (t-s)^{-3/4} s^{-3/4} ds.
\end{align}
The last integral is
$t^{-1/2} \int_0^1 (1-\sigma)^{-3/4} \sigma^{-3/4} d\sigma \equiv t^{-1/2} a_4$
for a constant $a_4$.
Therefore \eref{5.75} yields
\beq
 t^{3/4} \|B(t)\|_\infty \le c_N\Big\{ \|B_0\|_2
     + \beta(t) c \|B_0\|_2 t^{1/4} a_4   \Big\}  \ \text{for}\ 0< t \le 1.   \label{5.79}
      \eeq
 Replace $t$ by $t' <t$ in this inequality and take
the supremum over $t' <t$ to find, using the monotonicity of $\beta(t)t^{1/4}$,
\beq
 \beta(t) \le c_N\|B_0\|_2
     + \beta(t) \{t^{1/4} c \|B_0\|_2  c_N a_4\}   \ \text{for}\ 0< t \le 1.    \label{5.80}
      \eeq
Let $a= (2 c_N a_4)^{-1}$. Then the inequality \eref{AA0} may
be written
   \beq
   (2\tau)^{1/4} c\|B_0\|_2 c_Na_4  \le 1/2 .    \label{5.81}
   \eeq
   Hence, for $0 < t \le 2\tau \le 1$, \eref{5.80} yields
   $\beta(t) \le c_N \|B_0\|_2 +(1/2) \beta(t)$ and thus
   $ \beta(t) \le 2c_N \|B_0\|_{L^2(M)}$ if $ 0 < t \le  2\tau \le 1$.
   This proves the inequality  \eref{AA1}.

    Now \eref{AA2} follows from \eref{AA1} easily thus. Write $R = 2c_N \|B_0\|_2$.
         If $\tau \le t \le 2\tau$ then, from  \eref{AA1}, it follows that
    $\tau^{3/4} \|B(t)\|_\infty \le t^{3/4}  \|B(t)\|_\infty \le  R$,
    which is \eref{AA2} on the interval $[\tau, 2\tau]$.
    We may repeat the previous argument over an interval whose time origin is  $\tau$.
    Since  $\|B(\tau)\|_2 \le \|B_0\|_2$ the definition  \eref{AA0}
    shows that we
    we may take the ``new $\tau$''  to  be the same as the old $\tau$.
     Apply the inequality    \eref{AA1}    over the interval $(\tau, 3\tau]$.
         Taking $t$ in the second half of this interval, i.e. in $[2\tau, 3\tau]$, we find
     $\tau^{3/4} \|B(t)\|_\infty \le (t - \tau)^{3/4} \|B(t)\|_\infty \le 2c_N \|B(\tau)\|_2
    \le 2c_N \|B_0\|_2$,
     which is \eref{AA2} over the interval $[2\tau, 3\tau]$. Proceeding in this way,
     $\tau$ units at a time,  we find that  \eref{AA2} holds  over the whole interval $[\tau, \infty)$.

           Turning to the proof of \eref{AA3}, take $\w(s, x) = A'(s,x)$ in
   Proposition \ref{propDomination}    and take
   $h(s) = B(s)\# A'(s) - [ A'(s)\lrc B(s)]$
 over the interval $[0,2 \tau)$. Once again one needs to verify
  the boundary conditions  \eref{5.53} or \eref{5.54}.

        If $A(\cdot)$ satisfies Marini boundary conditions then
        $\w_{norm} = (A')_{norm} = - (d_{A}^* B)_{norm} = 0$ by \eref{M1}
        and \cite[Equ (3.20)]{CG1}.
             Moreover $(d_A \w)_{norm} = (d_A A')_{norm} = (B')_{norm} =0$.
        Therefore \eref{5.53} holds for $\w = A'$.
         Since Neumann boundary conditions are stronger than
        Marini boundary conditions,  \eref{5.53} holds in that case also.
             If $A(\cdot)$ satisfies Dirichlet boundary conditions then
       $\w_{tan} = (A')_{tan} =0$ by \eref{D1}. Moreover
       $d_A^*\w = d_A^* A' = - (d_A^*)^2 B =0$
       by  \cite[Equ (3.24)]{CG1}.
       In either case we may therefore apply
       Proposition \ref{propDomination}
        to the choice $\w= A'$.
 The inequality \eref{5.55} then gives the estimate
 \begin{align*}
 \|A'(t)&\|_\infty \le \| e^{t\Delta_N} A'(0) \|_\infty
    + \| \int_0^t e^{(t-s) \Delta_N} |B(s) \# A'(s) + A'(s) \lrc B(s)| ds \|_\infty \\
    &\le c_N\Big\{ t^{-3/4} \|A'(0)\|_2 + \int_0^t (t-s)^{-3/4} \|B(s) \# A'(s) + A'(s) \lrc B(s)\|_2 ds\Big\}.
    \end{align*}
  But, using \eref{AA1} combined with Lemma \ref{lemfa37} below, we have
  \begin{align*}
   \|B(s) \# A'(s) - [ A'(s) \lrc B(s)]\|_2 &\le 2c \|B(s)\|_\infty \|A'(s)\|_2 \\
&\le 4c_N s^{-3/4} c\|B_0\|_2 \|A'(s)\|_2 \\
 &\le 4c_N s^{-3/4} c\|B_0\|_2 \|A'(0)\|_2 e^{8c_Nc\|B_0\|_2 t^{1/4}}.
 \end{align*}
 Hence, for $0 < t \le 2\tau$, we find
 \begin{align*}
  t^{3/4}\|A'(t)\|_\infty&\le c_N\Big\{ \|A'(0)\|_2   \\
  &+ t^{3/4}4c_N c\|B_0\|_2 \| A'(0)\|_2 e^{8c_Nc\|B_0\|_2 t^{1/4}}
   \int_0^t (t-s)^{-3/4} s^{-3/4} ds \Big\} \\
   & = c_N\Big\{ \|A'(0)\|_2
   +4c_N c\|B_0\|_2 \| A'(0)\|_2 e^{8c_Nc\|B_0\|_2 t^{1/4}} t^{1/4} a_4\Big\}\\
   &\le \|A'(0)\|_2\Big( c_N
   +  4c_N^2 \{(2\tau)^{1/4} c \|B_0\|_2\}
                                       e^{8c_N \{c \|B_0\|_2 (2\tau)^{1/4}\}}  a_4\Big) \\
   &\le \|A'(0)\|_2 \Big( c_N
   +  4c_N^2 a e^{8c_N a}  a_4\Big).
  \end{align*}
  This proves \eref{AA3} with $\gamma = c_N
   +  4c_N^2 a e^{8c_N a}  a_4$.

We may apply \eref{AA3} beginning at time $\sigma \ge 0$ instead of time zero to find
\beq
(t-\sigma)^{3/4} \|A'(t)\|_\infty
\le \gamma \|A'(\sigma)\|_2\ \ \text{if}\ \ \sigma < t \le \sigma + 2\tau.
\eeq
In particular, if $ \sigma + \tau \le t$ then $(t-\sigma)^{3/4} \ge \tau^{3/4}$ and therefore
\beq
\tau^{3/4} \|A'(t)\|_\infty \le \gamma  \|A'(\sigma)\|_2\ \ \text{if}\ \ t -2\tau
                     \le \sigma \le t -\tau .                   \label{AA34}
\eeq
Keeping $t$ fixed and integrating the square of this inequality over the interval
$t -2\tau \le \sigma \le t -\tau$ we find
\beq
\tau \tau^{3/2} \|A'(t)\|_\infty^2
 \le \gamma^2 \int_{t-2\tau}^{t-\tau} \|A'(\sigma)\|_2^2 d\sigma. \label{AA35}
\eeq
In view of the bound $\int_0^\infty \|A'(\sigma)\|_2^2 d\sigma \le \|B_0\|_2^2$,
established in \cite[Equ (6.5)]{CG1},
the bound \eref{AA4} follows and at the same time
the integrability over $[0, \infty)$ of the integrand on the right of
\eref{AA35} proves \eref{AA5}.
  \end{proof}

          \begin{lemma}\label{lemfa37} Let $\psi_\infty (t) = 2c \int_0^t \|B(s)\|_\infty ds$.
     Then
\beq
\|A'(t)\|_2^2 + 2\int_0^t e^{\psi_\infty(t)-\psi_\infty (s)} \|B'(s)\|_2^2 ds
                    \le e^{\psi_\infty(t)} \|A'(0)\|_2^2.                        \label{A33}
\eeq
In particular, if \eref{AA1} holds, then
\beq
\|A'(t)\|_2 \le e^{8c_Nc \|B_0\|_2 t^{1/4}} \| A'(0)\|_2  \ \
                                            \text{for}\ \ 0 <t \le 2\tau \le 1.     \label{AA35}
\eeq

\end{lemma}
     \begin{proof} In the identity (see \cite[Equ (5.8)]{CG1})
$$
(d/ds) \|A'(s)\|_2^2 + 2\|B'(s) \|_2^2  = -2 ([A'(s) \wedge A'(s)], B(s)),
$$
     use the inequality
$ 2| ([A'(s) \wedge A'(s)], B(s))|   \le 2c \|B(s)\|_\infty \|A'(s)\|_2^2   $
 to dominate the last term.
 One arrives at
 $(d/ds) \|A'(s)\|_2^2 + 2\|B'(s)\|_2^2 \le 2c \|B(s)\|_\infty \|A'(s)\|_2^2. $
Hence
\beq
(d/ds) (e^{-\psi_\infty (s)} \|A'(s)\|_2^2) +2 e^{-\psi_\infty (s)} \|B'(s)\|_2^2 \le 0.  \notag
\eeq
Integrate from $0$ to $t$ to get
\beq
e^{-\psi_\infty(t)} \|A'(t)\|_2^2 - \|A'(0)\|_2^2
      +2\int_0^t e^{-\psi_\infty (s)} \|B'(s)\|_2^2 ds \le 0          \notag
\eeq
which gives \eref{A33}.
 Now if \eref{AA1} holds then
 \begin{align*}
 \psi_\infty(t) & \le 2c \int_0^t 2c_N \|B_0\|_2 s^{-3/4} ds
  = 16 cc_N \|B_0\|_2 t^{1/4}.
 \end{align*}
 Using just the first term in \eref{A33} we find therefore that \linebreak
 $\|A'(t)\|_2^2 \le e^{16c_Nc \|B_0\|_2 t^{1/4}} \| A'(0)\|_2^2 $,
 which is \eref{AA35}.
\end{proof}

         \begin{Remark}{\rm  Our proof of uniqueness of solutions to the Yang-Mills heat equation
\eref{ymh10} required   use of the allowed initial singularity of $\|B(t)\|_\infty$
specified in the definition \eref{ymh12} of  ``locally bounded''. As to whether uniqueness
holds without such an assumption, we have not been able to decide.
J. R{\aa}de,  \cite{Ra}, has proven uniqueness of solutions if one defines
a solution to be a limit of smooth solutions. The following corollary shows that
 such a limiting solution is automatically locally bounded and therefore
 our uniqueness  proof applies to such limiting solutions
   when $M$ is a bounded, smooth,   convex subset of $\R^3$.
}
\end{Remark}

                \begin{corollary} \label{corAA2} Suppose that, for some $T\le \infty$,
 $A(\cdot)$ is a strong solution on
$[0,T)$ satisfying Neumann, Dirichlet, or Marini boundary conditions.
Assume that there is a sequence  $A_n$ of smooth solutions on
 $[0,T)$, satisfying  the same boundary conditions as $A$,
  such that  $A_n \rightarrow A$ in the $C_{loc}( [0,T), W_1)$ topology.
 That is,
\beq
\sup_{0 \le t \le t_0} \| A_n(t) - A(t) \|_{W_1} \rightarrow 0 \ \ \text{for each}\ \ t_0 \in (0,T)
\eeq
Then $A(\cdot)$ is locally bounded.
\end{corollary}
          \begin{proof} Each function $A_n$ is clearly a locally bounded strong solution
on $[0,T)$. Since $\|A_n(0)\|_{W_1} $ is uniformly bounded in $n$ there exists
a constant $R >0$ such that  $\|B_n(0)\|_2 \le R$ for all $n$.
 By Theorem  \ref{thmAA} there exists $\tau >0$, depending only on $R$, such that
\beq
t^{3/4} \| B_n(t)\|_\infty \le 2c_N R\  \ \text{for } 0 < t \le 2\tau.
\eeq
For each $t >0$ the $W_1$ convergence of $A_n(t)$ to $A(t)$ implies that
$B_n(t) \rightarrow B(t)$ in $L^2(M)$. Hence
\beq
t^{3/4} \| B(t)\|_\infty \le 2c_N R \  \ \text{for } 0 < t \le 2\tau.
\eeq
The same argument applies on any interval
 $ [\alpha, \alpha  +2\tau] \subset [0,T)$ because $\|B_n(\alpha) \| _2 \le \| B_n(0)\|_2$.
 Therefore $
\| B(t)\|_\infty \le 2c_N R \tau^{-3/4}$ for $\alpha+\tau \le t \le \alpha +2 \tau$. Hence
$\|B(t)\|_\infty$ is bounded on any interval $[b,c] \subset [\tau, T)$.
\end{proof}

\begin{theorem} \label{thmAA3} Suppose that  $A(\cdot)$ is a locally
 bounded strong solution on an interval
$[0,T)$ satisfying Neumann or Dirichlet boundary conditions.
Then  \eref{AA1} and \eref{AA2} hold for a number $\tau$ depending
only on $\|B(0)\|_2$. Moreover if $\|A'(0)\|_2 <\infty$ then \eref{AA3} holds also.
\end{theorem}
        \begin{proof} For a given locally bounded strong solution $A(\cdot)$ we know
from the gauge invariant regularization lemma, \cite[Lemma 9.1]{CG1}
that if $T_0 <T$ then there exists $\epsilon >0$, depending on
 $\gamma \equiv \sup_{0\le s \le T_0}\|A(s)\|_{W_1}$,  such that,
  for any interval $[a,b] \subset (0, T_0]$
 of length at most $\epsilon$,  there is a
 sequence, $A_n$ of
 smooth solutions over $[a,b]$ which approximate $A$ over this interval in the strong
 sense given in \cite[Equ (9.1)]{CG1}.
 We need to modify the simple argument of
 Corollary \ref{corAA2} to take into account the possibility   that $\epsilon$, which
  depends on $\gamma$ and therefore on $\|A(\cdot)\|_{W_1}$ over the interval
  $[0,T_0]$, may be much smaller than the desired number $\tau$,
  which we hope will depend only on $\|B(0)\|_2$.
     To this end we will have to derive \eref{AA30} for non-smooth solutions
      to \eref{ymh10}.
 If $[a,b] \subset (0, T_0]$ is an interval of length at most $\epsilon$ and $A_n$
 denotes the sequence of smooth approximations of $A$ over $[a,b]$, then
 \eref{AA30} shows that
 \begin{align}
|B_n(b, x)| &\le \{e^{(b-a)\Delta_N} |B_n(a)|\}(x)
      + \int_a^b \{e^{(b-s)\Delta_N} |B_n(s) \# B_n(s)|\}(x) ds.   \notag
      \end{align}
Since $B_n$ converges to $B$ uniformly over $[a,b]\times M$ by
\cite[Equ (9.1)]{CG1},  and since
$e^{t\Delta_N}$ is bounded on $L^\infty(M)$, we may pass to the limit in
the last inequality to find
 \begin{align}
|B(b, x)| &\le \{e^{(b-a)\Delta_N} |B(a)|\}(x)
      + \int_a^b \{e^{(b-s)\Delta_N} |B(s) \# B(s)|\}(x) ds   \label{AA50}
      \end{align}
 for any interval $[a,b] \subset (0, T_0]$  of length at most $\epsilon$.
We will show in  Lemma \ref{lemAA4} that the validity of the pointwise inequality
\eref{AA50} over these small intervals implies its validity over large intervals.
Assuming then that \eref{AA50} holds over any interval  $[a,b] \subset (0, T_0]$
we will show that
   the derivation leading from \eref{AA30}   to \eref{5.79} now goes through exactly
   as before, provided we replace  the interval $[0,t]$ by the interval $[a,t]$, with
   the number $a$ necessarily greater
   than zero. Thus, defining
   $\beta_a(t) =\sup_{a\le s \le t}(s-a)^{3/4} \|B(s)\|_\infty$,
   the derivation of \eref{5.79} shows that
   \beq
   (t-a)^{3/4} \|B(t)\|_\infty
  \le  c_N \Big\{ \|B(a)\|_2 +  c \beta_a(t) \|B(a)\|_2 (t-a)^{1/4} a_4\Big\},  \label{AA53}
   \eeq
   for $a \le t \le T_0$. Now fix $t >0$ and let $a\downarrow 0$. Each term in
   \eref{AA53} converges to the corresponding term in \eref{5.79}.  Moreover
   the hypothesis that $A(\cdot)$ is locally bounded shows that
    $\beta(t) < \infty $ for all $t >0$. The remainder of the proof that \eref{AA1}
    holds is now exactly the same as in the proof of Theorem \ref{thmAA}. The proof of
    \eref{AA2} also follows as before.

          The proof of \eref{AA3} is similar: Taking $\w(t) = A'(t)$ in
     Proposition \ref{propDomination},
          one finds  the pointwise bound
 \beq
 |A'(b, x)|
 \le \{e^{(b-a)\Delta_N} |A'(a)|\}(x) +\int_a^b \{e^{(b-s)\Delta_N} |h(s)| \}(x) ds \label{AA54}
 \eeq
 for smooth solutions over  an interval $[a,b]$.        We may apply the gauge
  invariant regularization      lemma, \cite[Lemma 9.1]{CG1},
 to the given locally bounded
  strong solution $A(\cdot)$ and conclude that \eref{AA54} holds for small intervals
  $[a,b] \subset (0, T)$, and therefore, by the next lemma, holds
  for all intervals  $[a,b] \subset (0, T)$. The limiting procedure for letting
   $a\downarrow 0$,
  used in the proof of \eref{AA1},  applies now equally well to the proof of \eref{AA3}.
 \end{proof}

      \begin{lemma}\label{lemAA4} $($Time dependent semigroup inequality.$)$ Let $u(t,x)$ and $g(t,x)$ be non-negative bounded measurable functions on
 $[0,T) \times M$. Let $\epsilon >0$. Suppose that
 \beq
 u(b,x)
 \le \{e^{(b-a)\Delta_N} u(a,\cdot)\}(x) + \int_a^b \{e^{(b-s)\Delta_N}g(s)\}(x) ds
 \ \ \text{for}\ \ a.e.\ x    \label{AA60}
 \eeq
 whenever $ 0 < a < b <T$ and
 \beq
 b-a < \epsilon.                    \label{AA61}
 \eeq
 Then \eref{AA60} holds for all intervals $[a,b] \subset (0, T)$.
 \end{lemma}
       \begin{proof} Let  $0 < a < b < c <T$ and suppose
    $c-b < \epsilon$. For an induction proof, suppose that \eref{AA60} holds for
    this $a$ and $b$. Then
    \begin{align*}
    u(c) &\le e^{(c-b) \Delta_N} u(b) +\int_b^c e^{(c-s)\Delta_N} g(s) ds \\
    &\le e^{(c-b)\Delta_N} \Big\{e^{(b-a)\Delta_N} u(a) + \int_a^b e^{(b-s)\Delta_N}g(s) ds\Big\} + \int_b^c e^{(c-s)\Delta_N} g(s) ds \\
    &= e^{(c-a)\Delta_N}u(a)
                 +\int_a^b e^{(c-s)\Delta_N}g(s) ds +  \int_b^c e^{(c-s)\Delta_N} g(s) ds \\
    &=  e^{(c-a)\Delta_N}u(a)  + \int_a^c e^{(c-s)\Delta_N} g(s) ds.
    \end{align*}
    Therefore, given any interval $[a,b] \subset (0,T)$, one can partition it into small
    subintervals $ a =a_0 < a_1 < \cdots < a_n = b$ of length less than $\epsilon$
    and arrive at \eref{AA60} by induction.
    \end{proof}

\begin{Remark}{\rm  We have not included Marini boundary conditions in
the hypothesis of Theorem \ref{thmAA3} because the gauge invariant
regularization procedure used in the proof has not yet been proven
for Marini boundary conditions.
}
\end{Remark}


\section{Long time behavior} \label{secLTB}

It has been shown in several different contexts \cite{Ra,HT1,HT2} that
over a manifold without boundary,  a
solution to the Yang-Mills heat equation over $(0,\infty)$ converges to a limit
as  time goes to infinity through some sequence, if one counts only the gauge equivalence class at each time.
 Moreover, if one assumes a solution which is smooth for all time then the limiting
  connection is also gauge equivalent to a smooth connection \cite{Ra,HT1,HT2}, at least on an open dense set.
One can expect the same kind of behavior for a manifold with boundary.
In this section we are going to prove a  version of
such limiting behavior, but only in dimension three. It is aimed partly at showing how
 Wilson loop functions can be used to formulate such a convergence procedure
 and partly at showing how our gauge invariant regularization procedure smooths
 finite energy initial data enough to give meaning to such ``regularized Wilson loops''.

Given a connection  on a vector bundle, it is well known that the
associated  parallel transport operators along curves determine
the connection.  See e.g. \cite[Theorem 2.28]{Po}.
We are going to prove convergence of the  parallel transport operators
rather than convergence of the connection forms  themselves.
This is analogous to proving, for some sequence of unbounded self-adjoint
 operators $C_n$ on a Hilbert space, convergence of the unitary operators
 $e^{itC_n}$ instead of convergence of the $C_n$ themselves.

   Our main interest is in the regularization of rough gauge potentials, adequate for giving
   meaning to the Wilson loop function.
   We will begin with an example of a gauge potential with finite energy but
   which produces an infinite magnetic flux  through some loops.
   The Wilson loop function is meaningless for such loops.
   In the example we will  take the gauge  group to be the circle group.

 In   Section \ref{secLTB1}
we will review  how a parallel transport function on loops gives rise to a parallel
transport function on paths, with the help of homotopies. In Section \ref{secloops}
we will show  that, for a solution  to the Yang-Mills heat equation,
there is a sequence of times going to infinity for which the associated parallel
transport operators around loops converge.

\subsection{Magnetic field of a current carrying washer} \label{secwash}


\

A wire in $\R^3$ of zero thickness,  carrying current, produces a magnetic field of infinite energy.
We are going to describe a slightly smoother current distribution
which produces a magnetic field of finite energy and yet gives
  an infinite magnetic flux through certain loops. For such loops the Wilson
  loop functional is undefined.
  We will show in subsequent sections      how
   our  gauge invariant regularization procedure,  via the Yang-Mills heat
    equation, applies to  the Wilson loop function for finite energy gauge fields.

    Consider a washer  of zero thickness lying in the $x,y$ plane in $\R^3$ with center at the origin. We take
   the outer radius of the washer to be one and the inner radius to be $1/2$.
    A current circulates counterclockwise (viewed from above)
    through the washer in concentric circles centered at the origin.
    For a point $\ex$ on such  a circle, the current vector $\J(\ex)$ is tangent to the circle.
      See Figure \ref{fig:Wash7}.

\begin{figure}[htbp] 
   \centering
   \includegraphics[width=3in]{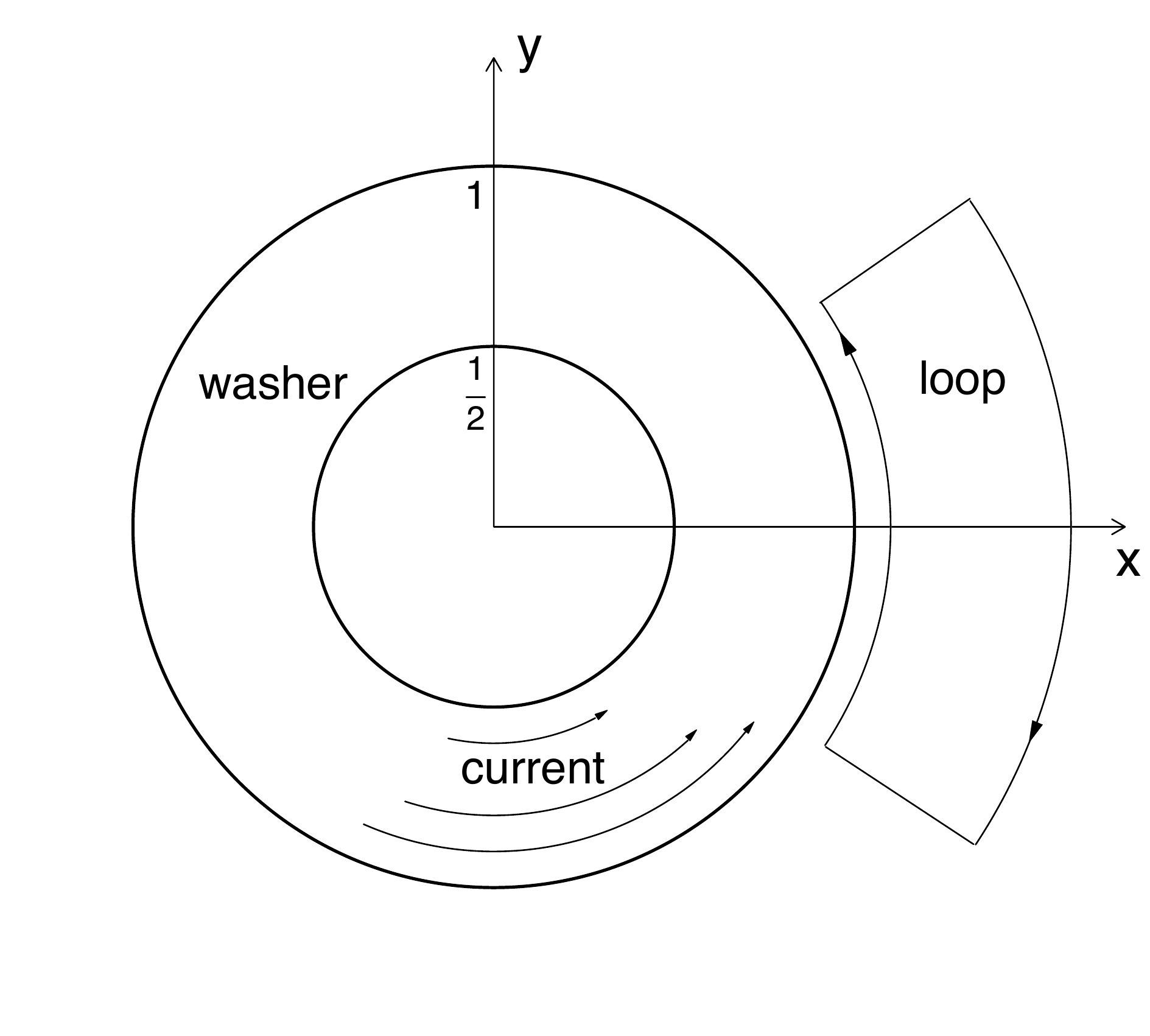}
   \caption{Current carrying washer}
   \label{fig:Wash7}
\end{figure}

    We take the current density to vary  with the distance from the origin and
    to be  heavily weighted toward the outer rim of the washer.
    We can write the planar current density explicitly     as
    \begin{align}
     \J(\ex) = \lambda(r) \Big( -\i \sin \phi + \j\cos \phi\Big)\ \ \text{when}\ \ \
      \ex = ( r \cos \phi, r \sin \phi, 0).
     \end{align}
   $\lambda(r)$ is the profile of the current strength as one moves from
    the inner rim at $r = 1/2$ to the outer rim at $r= 1$. By this we mean that $\lambda(r)dr$
    is the total current passing through a small radial interval $dr$ at distance $r$.
We will take
\begin{align}
\lambda(r)=\frac{1}{(1-r)(\log\frac{1}{1-r})^2},  \ \ \ 1/2 \le r < 1           \label{m7}
\end{align}
The intensity of current is therefore quite large near the outer rim. But the total current
circulating around the washer is $\int_{1/2}^1 \lambda(r) dr$, which is finite.
The magnetic potential produced by the current is given by $\Ae = (-\Delta)^{-1} \J$.
Thus
\begin{align}
4\pi \Ae(\ex) &= \int_{wash} \frac{\J(\ex')}{|\ex - \ex'|}d^2 \ex'       \notag \\
&=\int_{1/2}^1 dr \lambda(r) \int_{-\pi}^{\pi}
\frac{ -\i  \sin \phi + \j \cos \phi}{ |\ex - ( r \cos \phi, r \sin \phi, 0)|} d\phi              \label{m9}
\end{align}
The energy of this field is given by (See e.g. \cite[Equ. 5.153]{Jac}.)
\begin{align}
W = \int_{wash} \int_{wash} \frac{ \J(\ex) \cdot \J(\ex')}{|\ex - \ex'|} d^2\ex d^2\ex'    \label{m10}
\end{align}
where $\int_{wash}$ means the two dimensional integral over the washer. This is equivalent to the $H_1$ norm of $\Ae$ because
$\|\Ae\|_{H_1}^2  = \int_{\R^3}\sum_{j=1}^3 \p_j \Ae \cdot\p_j \Ae  d^3x
= (-\Delta \Ae, \Ae) = (\J, (-\Delta)^{-1} \J)$.

\begin{theorem}\label{washer}\

  1$)$ The gauge potential $\Ae$ has finite energy.

2$)$  There are piecewise smooth curves of finite length  in the $x,y$ plane
through which the  magnetic flux is infinite. In particular
the holonomy  (Wilson loop) $ W_C({\Ae}) = e^{i\infty}$ is undefined for such a curve $C$.
\end{theorem}
           \begin{proof}  To bound the energy \eref{m10}  take $ \ex = (r \cos \phi, r \sin \phi, 0)$ and
           $\ex ' = (r' \cos \phi', r' \sin \phi', 0)$  in \eref{m10}. Then
      $|\ex - \ex'|^2 = (r\cos \phi - r'\cos \phi')^2 + (r\sin \phi - r'\sin \phi')^2 = r^2 + (r')^2 - 2rr' \cos(\phi -\phi') =(r-r')^2 +2rr'(1-\cos(\phi -\phi'))$ and
 $\J(\ex)\cdot \J(\ex') =  \lambda(r) \lambda(r') (-\i \sin\phi + \j \cos\phi)\cdot(-\i \sin \phi' + \j \cos \phi') =  \lambda(r) \lambda(r')  \cos (\phi - \phi')$. Hence
 \begin{align}
W &= \int_{1/2}^1 \int_{1/2}^1 drdr' \lambda(r) \lambda(r')
  \int_{-\pi}^\pi \int_{-\pi}^\pi
            \frac{ \cos(\phi - \phi')\ d\phi d\phi'}{ \Big( (r-r')^2 +2rr'(1-\cos(\phi -\phi'))\Big)^{1/2}} \notag\\
  &=          2\pi  \int_{1/2}^1 \int_{1/2}^1drdr'  \lambda(r) \lambda(r')\int_{-\pi}^\pi
\frac{\cos \theta\ d\theta }{\Big((r-r')^2 + 2rr'(1- \cos\theta)\Big)^{1/2} }.  \label{m12}
 \end{align}

Since $\lambda(r)$ has singular behavior near $ r =1$ we will need to bound
the $\theta$ integral above  to prove finite energy. We will also need a lower bound
later to prove that $\Ae(\ex)$ is unbounded.
For these purposes we will show in Section \ref{seculbds} that
  there are strictly positive constants   $c_1, c_2, C_1, C_2$ such that
\begin{align}
c_1 + c_2 \log\frac{1}{u} &\le \int_{-\pi/4}^{\pi/4}
 \frac{\cos\theta}{\Big( u^2 + 2v^2 (1-\cos \theta)\Big)^{1/2}} d \theta
              \le C_1 + C_2  \log\frac{1}{u}                                   \label{m51}\\
&\text{for}\ \  0 < u <1\ \ \text{and}\ \  1/2 \le v \le 2        \label{m52}
\end{align}

 Now let $u = |r - r'|$ and $v = (rr')^{1/2}$. Then   \eref{m52} is satisfied for all $r$ and $r'$
 entering the integrals in \eref{m12}. The contribution
 to the $\theta$ integral in \eref{m12}  from $|\theta| \ge \pi/4$ is a bounded function of $ r$
 and $r'$ and since $\lambda(r)$ is integrable  the contribution to   \eref{m12} from
 $|\theta| \ge \pi/4$ is finite. In view of the second inequality in \eref{m51} it suffices therefore
  to show that
  \beq
  \int_{1/2}^1dr \int_{1/2}^1dr'  \lambda(r) \lambda(r') \log \frac{1}{|r-r'|}  < \infty.
  \eeq
  Since the possibly non-integrable singularity is near $r= r' =1$ it will be
   more perspicuous to change variables
  to $s = 1-r$ and $s' = 1 - r'$.   Thus we need to show that
  \begin{align}
  \int_0^{1/2} \int_0^{1/2} \mu(s) \mu(s') \log\frac{1}{|s - s'|} ds ds' < \infty      \label{m32}
  \end{align}
  when
  $\mu(s) = (s (\log s)^2)^{-1}$. The value of  this double integral over the two triangles
  $s \le s'$ and $ s' \le s$ is the same. So it suffices to show that one of them is finite.
  In fact we will show that
  \beq
  \int_0^{s'} \mu(s)\log\frac{1}{|s - s'|} ds
  \eeq
  is bounded for $ 0 \le s' \le 1/2$, which will prove \eref{m32} because $\mu(s')$ is integrable.

  Let $c = s'/2$. Now $\log\frac{1}{s' -s}$ is an increasing function of $s$ on $(0,c)$ while
  $\mu(s)$ is a decreasing function of $s$ on $(c, s')$. Hence
  \begin{align}
  \int_0^{s'} \mu(s) \log\frac{1}{s' -s} ds = \int_0^c \mu(s) \log\frac{1}{s' -s} ds + \int_c^{s'} \mu(s) \log\frac{1}{s' -s} ds  \notag\\
  \le \log\frac{1}{s' -c} \int_0^c \mu(s) ds + \mu(c) \int_c^{s'} \log\frac{1}{s' -s} ds    \label{m34}
  \end{align}
  Both integrals can be done explicitly. One finds
  $\int_0^c \mu(s) ds = (\log (c^{-1}))^{-1} =(\log (2/s'))^{-1}$
  and $\int_c^{s'} \log\frac{1}{s' -s} ds =(s'/2) ( 1 +\log (2/s'))$.
  Hence the right side of \eref{m34} equals
  \begin{align*}
  \Big(\log\frac{2}{s'}\Big) \frac{1}{\log\frac{2}{s'}}
  + \frac{1}{(s'/2)(\log\frac{2}{s'})^2} \cdot (s'/2) \Big(1+ \log\frac{2}{s'}\Big)
                   = 1 +\frac{1+\log\frac{2}{s'}}{(\log\frac{2}{s'})^2},
  \end{align*}
  which is bounded on $ 0 < s' \le 1/2$.  This proves Part 1) of Theorem \ref{washer}.

      For Part 2 we need to understand the behavior of the magnetic potential
  $\Ae(\ex)$ as $\ex$ approaches the outer rim of the washer.
  Because of the cylindrical symmetry it will suffice to do this when $\ex$ lies in the
 $x,z$ plane. In fact it suffices to consider just  $ \ex = (x_1,0, x_3)$ with $x_1\ge 0$.
  The distance from $\ex$ to   a current element is
   $ | \ex - (r \cos \phi, r\sin \phi, 0)|^2 = (x_1 - r\cos \phi )^2 + r^2 \sin^2 \phi + x_3^2
   = x_1^2 + x_3^2 +r^2 - 2 x_1 r \cos \phi = (x_1 - r)^2 + x_3^2 + 2x_1 r (1 -\cos\phi)$.
   Inserting this into \eref{m9} we see that the denominator is an even function of $\phi$.
   The contribution of $\i \sin \phi$ in the integral is therefore zero. Hence
   \beq
  4\pi {\Ae}(\ex) = \j \int_{1/2}^1 dr \lambda(r) \int_{-\pi}^\pi
           \frac{\cos \phi}{\Big( (x_1 -r)^2 + x_3^2 + 2x_1 r (1 - \cos \phi)\Big)^{1/2}} d\phi  \label{m43}
   \eeq
   for $\ex$ in the $x,z$ plane. From the cylindrical symmetry we see that $\Ae(\ex)$ is horizontal for all $\ex \in \R^3$ and in fact is tangent to the horizontal circle which is centered
    on the $z$ axis    and passes through $\ex$. (On the $z$ axis $\Ae(\ex)$ is zero, as
    one sees by putting $x_1=0$ in \eref{m43}.)

    Of course $\Ae$ is a smooth function on the complement of the closed washer
    because the denominator in \eref{m9}  is locally bounded away from zero there.
    We need only focus attention on the behavior of $\Ae(\ex)$ for $\ex$ in a
     small neighborhood of $(1,0,0)$ in the $x,z$ plane. For such $\ex$ the contribution
     to the integral  from points in the washer where $|\phi| \ge \pi/4$
     produces a smooth function of $x_1, x_3$ for $x_1 >0$. We therefore need only
      to analyze  the behavior of the function $f$ defined by
     \beq
     f(x_1, x_3) = \int_{1/2}^1 dr \lambda(r) \int_{|\phi| \le \pi/4}
     \frac{\cos \phi\  d\phi}{\Big( (x_1 -r)^2 + x_3^2 + 2 x_1 r(1- \cos \phi) \Big)^{1/2}} .
     \eeq
     The first inequality in \eref{m51} will give
     a lower bound on this integral just outside the outer rim of the washer as follows.
  Suppose that
 $1 \le x_1 \le  5/4$ and $|x_3| \le 1/2$. Let $u^2 = (x_1 - r)^2 + x_3^2$ and $ v^2 = x_1 r$.
 The reader can verify that \eref{m52} is satisfied.
  Hence
     \begin{align}
     f(x_1, x_3) \ge  \int_{1/2}^1 \lambda(r)(c_1 + c_2 \log \frac{1}{u}) dr .
          \end{align}
    If $x_1 \downarrow 1$
     and $|x_3| \downarrow 0$ then $u\downarrow 1-r$ and   the monotone convergence theorem shows that,
     for some finite constant $C_6$, one has
     \begin{align}
   \liminf_{x_1\downarrow 0\  |x_3|\downarrow 0}  f(x_1,x_3) & \ge c_2 \int_{1/2}^1 dr \lambda(r) \log\frac{1}{1-r}  + C_6  \notag\\
     &=c_2\int_{1/2}^1 dr \frac{1}{(1-r)(\log\frac{1}{1-r})^2} \log\frac{1}{1-r} + C_6   \notag\\
     &= \infty.                                              \label{m48}
     \end{align}
     Therefore $\Ae(\ex)\cdot \j$ is infinite at the point $(1,0,0)$ and also goes
      to $\infty$ as $\ex \rightarrow (1,0,0)$ in the $x,z$ plane.

Consider now the loop shown in Figure \ref{fig:Wash7}. It is given by a closed curve $C$
lying in the $x,y$ plane and forming the boundary of an annular sector centered at the origin and whose inner radius is one. In Figure \ref{fig:Wash7} the curve is shown separated from
the rim of the washer for clarity. But we are interested in the circumstance in which the inner circle
of the annular sector coincides with a portion of the outer rim of the current carrying washer.
The outer circle segment of $C$
 is concentric with the inner one and is joined to it by radial lines.
   Since $\Ae$ is tangential to the outer rim
 of the washer it is also tangential to the inner circle of $C$.
 However this tangential component of $\Ae$ is infinite, as we have seen.
 Thus the integral of $\Ae$ along the inner circle of $C$
 is infinite.  The integral of
 $\Ae$ along the two radial lines
 is zero because $\Ae$ is perpendicular to these radial lines.
  The integral of $\Ae$ along the outer circle is finite because $\Ae$ is smooth
 in the vicinity of the outer circle. Thus $\int_C \Ae(\ex) \cdot d\ex = \infty$.
 This proves Part 2) of Theorem \ref{washer}.

        There is another sense in which this loop integral is infinite:  keep the outer circle
        of the curve $C$  fixed
        and shift the inner circle
        away from the outer rim of the washer by a small amount,
        say $\epsilon >0$, as is shown in Figure \ref{fig:Wash7}.
        For this curve $C_\epsilon$  the contour integral $ \int_{C_\epsilon} \Ae(\ex) \cdot dx$
        is finite, but  increases to infinity as $\epsilon \downarrow 0$ because, as   \eref{m48}
        shows, for $x_3=0$ the tangential component of $\Ae(\ex)$ increases to $\infty$
        as  $x_1 \downarrow 1$. In particular, by Stokes' theorem, the magnetic flux through
        the planar surface  bounded by $C_\epsilon$ increases
            to $\infty$ as $\epsilon \downarrow 0$. (The right hand rule
  also shows that the   magnetic field $B$ points downward everywhere on the ring.)
 \end{proof}

 \begin{remark}\label{remrestrict}{\rm  Trace
 theorems assert that   a function lying in a Sobolev space
 $H_s$ over a manifold  will, upon restriction to a co-dimension one submanifold $N$,
 lie in $H_{s -(1/2)}(N)$. This theorem notoriously breaks down if $s = 1/2$.
 In our case
 the magnetic field $\Ae$ lies in $H_1(\R^3)$ and therefore restricts to a function
  in $H_{1/2}(N)$ for any reasonable surface $N$. One can try to restrict it once more to
  a curve $C$ contained in $N$ and ask what properties the restriction to $C$ has.
  Since $s$ is now equal to $1/2$ the trace  theorem breaks down. One can not infer from it
  that the   restriction  to the curve $C$ is  an almost everywhere finite
  function. Our example shows that the very
   worst can happen: the restriction of the magnetic field to an arc of the outer rim of the washer is
    identically infinite. For further discussion of trace theorems see \cite[Theorem 9.5]{LM}
    and \cite{EL}.
 }
 \end{remark}

\subsubsection{Upper and lower bounds for an integral} \label{seculbds}

\begin{lemma}
 Let $0 < \theta_0 < \pi/2$ and let
$a^2 = (\sin\theta_0)/\theta_0$ with $a >0$. Then, for $u>0$ and $v >0$, there holds
\begin{align}
\frac{1}{v} \log (1 + \frac{v\theta_0}{u})
   \le \int_0^{\theta_0} \frac{1}{\Big( u^2 + 2 v^2 (1 - \cos\theta)\Big)^{1/2}} d \theta
   \le \frac{\sqrt{2}}{va}  \log (1 + \frac{va\theta_0}{u}) .        \label{m25}
   \end{align}
   In particular \eref{m51} holds.
\end{lemma}
\begin{proof} For $s \ge 0$ the inequalities $ 1+s^2 \le (1+s)^2 \le 2(1+s^2)$ imply that
\begin{align}
\int_0^b (1+s^2)^{-1/2} ds \ge \log(1+b)
\ge 2^{-1/2} \int_0^b (1+s^2)^{-1/2} ds\ \text{for}\  b >0.      \label{m26}
\end{align}
Since $a^2$ is the slope of a line segment lying below $\sin(\cdot)$, integration gives
$ a^2 \theta^2 \le 2(1-\cos\theta) \le \theta^2$ for $0 \le \theta \le \theta_0$. Hence
\begin{align*}
\int_0^{\theta_0}\frac{d\theta}{\Big( u^2 + v^2 a^2 \theta^2\Big)^{1/2} }
\ge \int_0^{\theta_0}\frac{d\theta}{\Big(u^2 + 2v^2 ( 1- \cos\theta)\Big)^{1/2}}
\ge \int_0^{\theta_0}\frac{d\theta}{\Big( u^2 + v^2  \theta^2\Big)^{1/2} }
\end{align*}
Change variables in the left-most integral to $s = (va/u)\theta$ to find
$(va)^{-1}\int_0^{va\theta_0/u} (1+ s^2)^{-1/2} ds$, which, by \eref{m26}, is at most
 $(va)^{-1}\sqrt{2} \log(1 + (va\theta_0/u))$. The other half of \eref{m25} follows similarly.

          For the proof of \eref{m51} we can ignore the factor $\cos\theta$ in the numerator
           of \eref{m51}  because it is bounded and  bounded away from zero  on $[-\pi/4, \pi/4]$.
           We are assuming now that $1/2\le v \le 2$, from which
   it follows that the left side of \eref{m25} dominates the left side of \eref{m51} for some constants
   $c_1, c_2$ and for all $u \in (0,1)$. Moreover, since  $va\pi/4 \le2$, the right side of
   \eref{m25} is at most $(2\sqrt{2}/a)\log(1 +(2/u)) \le (4\sqrt{2}/a)\log(2/u)$  because
   $\log (1+x) \le 2 \log x$ when $x \ge 2$.
\end{proof}

\subsection{From loops to paths} \label{secLTB1}

\begin{notation}\label{notpaths}
{\rm  Let $M$ be the closure of a bounded open set in $\R^3$
with smooth convex boundary.  Denote by $\Gamma$ the set of piecewise
$C^1$ functions from $[0,1]$ into $M^{int}$.
If $\gamma$ and $\mu$ are two elements of $\Gamma$
such that $\gamma(1) = \mu(0)$ then their concatenation $\gamma\mu$ is defined by
 \beq
 (\gamma \mu)(s) =
 \begin{cases} \gamma(2s), &0\le s \le 1/2 \\
                        \mu(2s -1), &1/2 \le s \le 1.
   \end{cases}             \label{ltp3}
   \eeq
 The curve $\gamma\mu$ is clearly again in $\Gamma$.
 The inverse path is defined as usual
 by $\gamma^{-1} (s) = \gamma(1-s), 0 \le s \le 1$. The path $\gamma \gamma^{-1}$
  retraces itself.

      By a {\it parallel transport system} in the bundle $\V\times M \rightarrow M$
      we mean a map
      $\Gamma \ni \gamma \mapsto //_\gamma \in  End\ \V$ such that

 i) $//_{\gamma\circ \phi} = //_\gamma$    for any homeomorphism $\phi:[0,1]\rightarrow [0,1]$, which, together with its inverse, is piecewise $C^1$,

 ii) $//_{\gamma\mu} = //_\gamma\ \  //_\mu$,

 iii) $//_{\gamma \gamma^{-1}} = I_\V$.

 Taking $\epsilon_0$ to be the trivial curve, $\epsilon_0(s) \equiv x_1$, it  follows from
 ii) and iii) that  $//_{\epsilon_0} = I$ and $(//_\gamma)^{-1} = //_{\gamma^{-1}}$.

Under further technical assumptions such a parallel transport system always
comes from a connection on the bundle $\V\times M \rightarrow M$. This has been discussed for example in \cite[Theorem 2.28]{Po}.

We are going to show that, for a solution $A(\cdot)$ to the Yang-Mills heat equation,
and for any sequence $t_k$ going to infinity,
the parallel transport operators $//_\gamma^{A(t_k)}$ converge to such a parallel transport
system   after suitable gauge transformations.

Choose  a point $x_0 \in M^{int}$ and denote the set of loops at $x_0$ by
      \beq
\Gamma_0 = \{ \gamma \in \Gamma : \gamma(0) = \gamma(1) = x_0 \}.   \label{ltp5}
      \eeq
A parallel transport system can be recovered, up to gauge transformation,
from its restriction to $\Gamma_0$
 by choosing a homotopy of $M$ with $x_0$.
 The well known procedure for doing this will be described
 in the following algebraic lemma.

Let $X$ be a manifold and let $x_0$ be a point in $X$. By a piecewise $C^1$ homotopy of
$X$ with $\{x_0\}$ we mean a continuous map
 $h: [0,1]\times
 X \rightarrow X$ with  $h(0,x) = x_0$, $ h(1,x) = x$,
 and, for all $x \in X$, the curve $s \mapsto h_x(s) := h(s,x)$ is
  piecewise $C^1$. We will assume also that $h(s, x_0)= x_0$ for all $s \in [0,1]$.
  Our limit results can easily be extended to non-contractible manifolds,
  but the analytic idea is already well illustrated in the contractible case,
  to which we will restrict our attention.

 \begin{lemma}  \label{loops-paths}
      Let  $X$ be a finite dimensional pathwise connected manifold. Let
      $x_0 \in X$. Denote by $\Gamma$ the set of piecewise $C^1$ functions
      from $[0,1]$ into $X$ and define
      $\Gamma_0 = \{ \gamma \in \Gamma : \gamma(0) = \gamma(1) = x_0 \}$.
      Suppose that $P: \Gamma_0 \rightarrow End\ \V$ is a map
      with the following properties

      $1)\ ($parametrization invariance$)$  $P(\gamma) = P(\gamma\circ \phi)$ for any
      piecewise $C^1$ homeomorphism $\phi: [0,1]\rightarrow [0,1]$ with
      piecewise $C^1$ inverse.

      $2)$  $P(\gamma \mu) = P(\gamma) P(\mu)$ for all $\gamma$ and
      $\mu$ in $\Gamma_0$.

      $3)$  $P(\gamma \gamma^{-1}) = I_\V$ for any path $\gamma \in \Gamma$ with
      $\gamma(0) = x_0$.

   \noindent
      Let $h:[0,1]\times X \rightarrow X$ be a piecewise $C^1$ homotopy of $X$ to $x_0$.

      Then there is a  unique parallel transport system
      $//_\gamma,\ \gamma \in \Gamma$, which is the identity along all homotopy
      paths $h_x(\cdot)$ and agrees with $P$ on $\Gamma_0$.
      \end{lemma}
                    \begin{proof}  Suppose that $\gamma \in \Gamma$ with
   $\gamma(0) = x$ and  $\gamma(1) = y$.  Then the  path
      $h_x \gamma h_y^{-1}$  lies in $\Gamma_0$.
      Define $//_\gamma = P(h_x \gamma h_y^{-1})$.
   If $\mu \in \Gamma$ also and $\gamma(1) = \mu(0)$ and $\mu(1) =z$  then
   \beq
   //_{\gamma \mu} = P(h_x \gamma \mu h_z^{-1})
   = P((h_x \gamma h_y^{-1})(h_y \mu h_z^{-1})) = //_\gamma \ //_\mu
   \eeq
   by 2).  Moreover $//_{\gamma \gamma^{-1}}
   = P((h_x\gamma)( \gamma^{-1}h_x^{-1}))
   = P(h_x\gamma)(h_x\gamma)^{-1})= I_\V$ by  3).
   Thus items ii) and iii) are verified. Item i) is clear. Moreover
   $//_{h_x} = P(h_{x_0}h_x (h_x)^{-1}) = I_\V$.  For any other parallel transport system
  $///$ with the stated properties one has  $///_\gamma = ///_{h_x}\  ///_\gamma\ ///_{h_y^{-1}}
  =///_{h_x \gamma h_y^{-1}} = P(h_x \gamma  h_y^{-1}) =//_\gamma$.
   \end{proof}

}
\end{notation}

   \subsection{Convergence on loops} \label{secloops}
\begin{notation}{\rm  Let $M$ be the closure of a bounded open set in $\R^3$
with smooth convex boundary.  Let $x_0 \in M^{int}$ and define
$\Gamma $ and $\Gamma_0$ as in Notation \ref{notpaths}.
In our simple setting a tangent vector to $M$ at a point $\gamma(s)$ is just
a vector $u(s) \in \R^3$. We will denote by $T_\gamma(\Gamma)$ any
piecewise $C^1$ function $u:[0,1]\rightarrow \R^3$. If $\gamma \in \Gamma_0$
then we will write $u \in T_{\gamma}(\Gamma_0)$ if $u$ is in piecewise
$C^1([0,1]; \R^3)$  and $u(0) = u(1) =0$. For $u \in T_\gamma (\Gamma)$
define
\beq
\| u \| = \sup_{0\le s \le 1} |u(s)|_{\R^3} + \sup_{0\le s \le 1} |u'(s)|_{\R^3}. \label{LTB1}
\eeq
For a curve $[a,b] \ni t\rightarrow \gamma_t(\cdot) \in \Gamma$
we take its length to be
$\int_a^b \|\p_t \gamma_t\| dt$ as usual and define the distance  $d_1(\gamma_1, \gamma_2)$
to be the infimum of lengths of curves joining $\gamma_1$ to $ \gamma_2$
in the manifold $\Gamma$.
$\Gamma$ and $\Gamma_0$ are (incomplete) metric spaces  in this metric.
}
\end{notation}

\begin{definition}{\rm For a smooth End $\V$ valued
connection form $A$ on $M^{int}$ and a piecewise
$C^1$ path $\gamma$ in $M$ the parallel transport operator along $\gamma$
is defined by the solution to the ordinary differential equation
\beq
g(t)^{-1} dg(t)/dt = A\<  d \gamma(t)/dt \>,\ \ \    g(0) = I_{\V}.
\eeq
We put $//_\gamma^A = g(1)$. Properties i), ii), iii) of Notation \ref{notpaths}
are well known for this map.
}
\end{definition}

In this section we are going to prove that for any locally bounded strong solution
of the Yang-Mills heat equation satisfying  Neumann or Dirichlet boundary conditions,
 and for any sequence of times
going to infinity,
there is a subsequence $t_j$ and gauge transforms $k_j$
such that the connection forms $A(t_j)^{k_j}$ are smooth and the
parallel transport operators $//_\gamma^{A(t_j)^{k_j}} $
converge, as operators from $\V$ to $\V$, to a map $P$ on $\Gamma_0$
satisfying all the conditions listed in Lemma \ref{loops-paths}.

\begin{theorem} \label{LTB} Suppose that $M$ is a compact convex subset of $\R^3$
with smooth boundary.
Let $A(\cdot)$ be a locally bounded strong solution of the
 Yang-Mills heat equation \eref{ymh10} over $[0,\infty)$  satisfying Dirichlet or
 Neumann boundary conditions.  Choose $x_0 \in M^{int}$.
 Suppose that $\{t_k\}$ is a sequence of times going to $\infty$.

 There is a function $P:\Gamma_0 \rightarrow End\ \V$ satisfying
 conditions {\rm 1), 2), 3)}
 of Lemma \ref{loops-paths},
     a subsequence $t_j$
and functions  $k_j \in W_1(M;K)$ such that

a$)$\ \  $k_j^{-1} dk_j \in W_1(M;\kf)$ for all $j$,

b$)$\ \  $\alpha_j \equiv A(t_j)^{k_j}$ is in $C^\infty(M;\L^1\otimes \kf)$ and,

c$)$ for each  $\gamma \in \Gamma_0$
the operators $//_\gamma^{\alpha_j}$
converge to $P(\gamma)$  as $j\rightarrow \infty$.

\noindent
Moreover

d$)$ $P$ is continuous on $\Gamma_0$ in the metric $d_1$.

In particular, given a piecewise $C^1$ homotopy of $M^{int}$ onto $x_0$,
there is a parallel transport system on $\Gamma$ that extends $P$.
\end{theorem}

\begin{Remark}{\rm If $\gamma$ is a closed curve in $M^{int}$
beginning at $x_0$, and
$A$ is a smooth connection form, then for any smooth  function $k: M \rightarrow K$
one has the well known identity.
\beq
//_\gamma^{A^k} = k(x_0)^{-1} (//_\gamma^A)k(x_0)
\eeq
Consequently
\beq
trace\ //_\gamma^{A^k} = trace //_\gamma^A.
\eeq
The function $A \mapsto  trace\ //_\gamma^A$ is therefore fully gauge invariant
and in particular is independent of the choice of gauge transformation
 $k$.  Theorem \ref{LTB} implies then that there exists a sequence of times
going to infinity for which the functions  $\;trace //_\gamma^{A(t_j)}$ converge for all piecewise
 $C^1$ loops $\gamma$ starting at $x_0$. One need not specify gauge transformations $k_j$ for this convergence.
 }
\end{Remark}

The proof of Theorem \ref{LTB} depends on the following lemmas.

           \begin{lemma} \label{lem1}
            Let $A(\cdot)$ be a locally bounded strong solution satisfying Dirichlet or Neumann
 boundary conditions and let $t_1 >0$.
Then there exists a continuous function $k: M \rightarrow K$ such that

a$)$ $k^{-1} dk \in W_1(M)$ and

b$)$ $\alpha\equiv A(t_1)^{k} \in C^\infty(M; \L^1\otimes \kf)$.
\end{lemma}
       \begin{proof}  The proof depends heavily on results in \cite{CG1}.
       From \cite[Corollary 9.3]{CG1}
  it follows that
        $\sup_{0<t \le t_1} \|A(t)\|_{H_1} < \infty$. Therefore, by \cite[Theorem 2.13]{CG1} there exists $T>0$ such
that, for any $t_0 \in (0, t_1)$, the parabolic equation
\beq
(\p/\p t)C = -(d_C^* B_C + d_C d^*C), t >0, \ \
         C(0) = A_0 .                                             \label{ST11}
\eeq
\cite[Equ (2.14)]{CG1} has a solution $C(\cdot)$ on the interval $[t_0, t_0 + T]$,
with $C(t_0) = A(t_0)$.  Pick $t_0 \in (0,t_1)$ such that $t_1 < t_0 + T$. \cite[Corollary 8.4]{CG1}
then ensures that there exists a continuous function $ g: M\rightarrow K$ such that $g^{-1} dg \in W_1$ and for which   $A(t_1) = C(t_1)^g$.  Since $C(t_1) \in C^{\infty}$ we may take $k = g^{-1}$.
        Take note here that the equality  $A(t_1) = C(t_1)^g$  relies on the uniqueness
        theorem, \cite[Theorem 8.15]{CG1},  which is applicable to the
         restriction of  $A(\cdot)$ to $[t_0,t_1]$.
    \end{proof}

\begin{lemma} \label{lem2}
Let $\gamma : [0,1]\rightarrow M$ be a piecewise $C^1$ closed curve starting at $x_0$.
Let $ u:[0,1]\rightarrow T(M)$ be a $C^1$ vector field along $\gamma$ for which
$u(0) = u(1) =0$.
That is, $u(s) \in T_{\gamma(s)}(M), 0 \le s \le 1$.
Let $\alpha$ be a smooth connection form on $M$ with bounded curvature $B$.
Then
\beq
\| \p_u //_\gamma^\alpha \|_{End \ \V}
 \le \| B \|_\infty  \sup_{0\le s \le1} |u(s)|\ \  \text{Length}(\gamma). \label{LTB5}
\eeq
\end{lemma}
        \begin{proof} Since $u(0) = u(1) = 0$ the identity  \cite[Equ (2.6)]{G2}
  shows that
\begin{align*}
\Big\|\p_u //_\gamma^\alpha\Big\|_{End \ \V}
&= \Big\|\int_0^1 //_{\gamma|_0^s}^\alpha
       \< B(\gamma(s)), \gamma'(s) \wedge u(s) \> ds \Big\|_{End \ \V}\\
 & \le \int_0^1 \| //_{\gamma|_0^s}^\alpha  \|_{End \ \V}
                    \|B \|_\infty |\gamma'(s) \wedge u(s) \> |_{\Lambda^2(\R^3)}     ds   \\
 & \le \|B \|_\infty  \int_0^1|\gamma'(s)| | u(s)| ds   \\
 &  \le \|B \|_\infty (\sup_{0 \le s \le 1} |u(s)| ) \int_0^1 |\gamma'(s)| ds,
\end{align*}
which is \eref{LTB5}.
\end{proof}

\bigskip
\noindent

\begin{proof}[Proof of Theorem \ref{LTB}]
 For each $\;t \ge 1$ we have constructed a gauge
  function  $\;k(t): M \rightarrow K$
such that $\alpha(t) \equiv A(t)^{k(t)}$ is a $C^\infty$ connection form. Denote by $B_{\alpha}(t)$ the curvature of the connection $\alpha(t)$.
 Let $\gamma$  and $\eta$ be in $\Gamma_0$ and of length at most $L$.
Define $u(s) = \gamma(s) - \eta(s)$ and let
 $\gamma_\sigma(s) = \eta (s)  + \sigma u(s)$.
        Then $\gamma_\sigma$ lies in $M^{int}$ for small $\sigma$
  and $\p_\sigma \gamma_\sigma = u$.  Let $b = \sup_{t\ge 1} \|B(t)\|_\infty$.
 We know that $ b < \infty$ by Theorem \ref{thmAA}. Since
$ \| B_{\alpha}(t)\|_ \infty = \| B(t)\|_\infty \le b$,
   Lemma  \ref{lem2} shows that
          \begin{align*}
\|\p_\sigma //_{\gamma_\sigma}^{\alpha(t)} \|_{End\ \V}
&\le b \sup_{0\le s \le 1}|\gamma(s) - \eta(s)|\ \cdot \text{Length}(\gamma_\sigma) \\
&\le b \sup_{0\le s \le 1}|\gamma(s) - \eta(s)| \
          \cdot [\text{Length}(\gamma) + \text{Length}(\eta) ] \\
    &\le 2bL  \sup_{0\le s \le 1}|\gamma(s) - \eta(s)|.
           \end{align*}
Hence
      \begin{align}
\| //_\gamma^{\alpha(t)} - //_{\eta}^{\alpha(t)} \|_{End\ V}
&\le \int_0^1 \| \p_\sigma //_{\gamma_\sigma}^{\alpha(t)} \|_{End\ V} d\sigma \notag \\
&\le 2bL  \sup_{0\le s \le 1}|\gamma(s) - \eta(s)| \label{LTB10}.
       \end{align}
    An Arzela-Ascoli type diagonalization argument
 shows that a pointwise \linebreak
   bounded, equicontinuous  sequence of
 functions on a separable metric space $S$ to a compact subset of $End \ V$
 contains a subsequence that converges
 pointwise to a continuous function (and of course the convergence is uniform on compact subsets.) Taking the metric space to be the set
 $\mathcal C_L \equiv \{\gamma \in \Gamma_0: \text{Length}( \gamma)  \le L\}$
  with the metric
 $d_0(\gamma, \eta) = \sup_{0\le s \le1} |\gamma(s) - \eta(s)|$, and taking the
 functions to be $\{ //_\gamma^{\alpha(t)}\}$  with ranges
  contained in $K \subset End\ V$,
 the  estimate \eref{LTB10}  shows that we may apply this Arzela-Ascoli
 argument and conclude that for any sequence of times increasing to $\infty$
there is a subsequence $t_j\uparrow \infty$
 for which $//_{\gamma}^{\alpha(t_j)}$ converges in operator norm for each
 curve $\gamma \in \mathcal C_L$.
  We may allow $L \uparrow \infty$ through a sequence and use
  diagonalization again to conclude
 that there is a function $ P:\Gamma_0 \rightarrow End\ V$
 such that $P(\gamma) = \lim_{j\to \infty} //_\gamma^{\alpha(t_j)}$ in operator norm
 for all $\gamma \in \Gamma_0$.  By \eref{LTB10}  $P|\mathcal C_L$ is continuous
 in the norm $d_0$ for each $L <\infty$ and therefore is continuous on
  $\Gamma_0$  in the metric $d_1$. The properties 1), 2), 3) of
  Lemma \ref{loops-paths} follow from the corresponding properties of the
  maps $\gamma \mapsto //_{\gamma}^{\alpha(t_j)}$. The map $P$
  therefore extends, by Lemma \ref{loops-paths}, to a parallel transport
  system on paths when a piecewise $C^1$ homotopy of $M^{int}$ to $x_0$
  is specified. The extension is unique in the sense given in Lemma \ref{loops-paths}.
 \end{proof}

\begin{acknowledgement}
N. Charalambous would like to thank the Asociaci\'{o}n Mexicana de Cultura A.C.
  \end{acknowledgement}


\begin{bibdiv}
\begin{biblist}


\bib{Bal}{article}{
    AUTHOR = {Balaban, T.},
     TITLE = {Convergent renormalization expansions for lattice gauge
              theories},
   JOURNAL = {Comm. Math. Phys.},
  FJOURNAL = {Communications in Mathematical Physics},
    VOLUME = {119},
      YEAR = {1988},
    NUMBER = {2},
     PAGES = {243--285},
      ISSN = {0010-3616},
     CODEN = {CMPHAY},
   MRCLASS = {81E25 (81E08 81E15)},
  MRNUMBER = {968698 (90b:81104)},
MRREVIEWER = {Claus Montonen},
       URL = {http://projecteuclid.org/getRecord?id=euclid.cmp/1104162401},
}



\bib{Cha4}{article}{
    AUTHOR = {Charalambous, Nelia},
     TITLE = {Eigenvalue estimates for the {B}ochner {L}aplacian and
              harmonic forms on complete manifolds},
   JOURNAL = {Indiana Univ. Math. J.},
  FJOURNAL = {Indiana University Mathematics Journal},
    VOLUME = {59},
      YEAR = {2010},
    NUMBER = {1},
     PAGES = {183--206},
      ISSN = {0022-2518},
     CODEN = {IUMJAB},
   MRCLASS = {58J60 (35P05 35R01 58J50)},
  MRNUMBER = {2666477 (2011j:58060)},
MRREVIEWER = {Julie Rowlett},
       DOI = {10.1512/iumj.2010.59.3770},
       URL = {http://dx.doi.org/10.1512/iumj.2010.59.3770},
}

\bib{CG1}{article}{
    AUTHOR = {Charalambous, Nelia},
     AUTHOR = {Gross, Leonard},
     TITLE = {The {Y}ang-{M}ills heat semigroup on three-manifolds with
              boundary},
   JOURNAL = {Comm. Math. Phys.},
  FJOURNAL = {Communications in Mathematical Physics},
    VOLUME = {317},
      YEAR = {2013},
    NUMBER = {3},
     PAGES = {727--785},
      ISSN = {0010-3616},
     CODEN = {CMPHAY},
   MRCLASS = {58J35 (58J32 81T13)},
  MRNUMBER = {3009723},
MRREVIEWER = {Thomas Krainer},
       DOI = {10.1007/s00220-012-1558-0},
       URL = {http://dx.doi.org/10.1007/s00220-012-1558-0},
}
	


\bib{CFKS}{book}{
    AUTHOR = {Cycon, H. L.},
     AUTHOR = {Froese, R. G.},
      AUTHOR = { Kirsch, W. },
      AUTHOR = {Simon, B.},
     TITLE = {Schr\"odinger operators with application to quantum mechanics
              and global geometry},
    SERIES = {Texts and Monographs in Physics},
   EDITION = {Study},
 PUBLISHER = {Springer-Verlag},
   ADDRESS = {Berlin},
      YEAR = {1987},
     PAGES = {x+319},
      ISBN = {3-540-16758-7},
   MRCLASS = {35-02 (35J10 47F05 58G40 81C10)},
  MRNUMBER = {MR883643 (88g:35003)},
MRREVIEWER = {M. Demuth},
}


\bib{Do1}{article}{
    AUTHOR = {Donaldson, S. K.},
     TITLE = {Anti self-dual {Y}ang-{M}ills connections over complex
              algebraic surfaces and stable vector bundles},
   JOURNAL = {Proc. London Math. Soc. (3)},
  FJOURNAL = {Proceedings of the London Mathematical Society. Third Series},
    VOLUME = {50},
      YEAR = {1985},
    NUMBER = {1},
     PAGES = {1--26},
      ISSN = {0024-6115},
     CODEN = {PLMTAL},
   MRCLASS = {58E15 (14F99 53C05 57R99)},
  MRNUMBER = {MR765366 (86h:58038)},
MRREVIEWER = {S. Ramanan},
}



\bib{EL}{article}{
    AUTHOR = {Einav, Amit},
    AUTHOR = {Loss, Michael},
     TITLE = {Sharp trace inequalities for fractional {L}aplacians},
   JOURNAL = {Proc. Amer. Math. Soc.},
  FJOURNAL = {Proceedings of the American Mathematical Society},
    VOLUME = {140},
      YEAR = {2012},
    NUMBER = {12},
     PAGES = {4209--4216},
      ISSN = {0002-9939},
     CODEN = {PAMYAR},
   MRCLASS = {35A23 (26D10)},
  MRNUMBER = {2957211},
MRREVIEWER = {Jean Van Schaftingen},
       DOI = {10.1090/S0002-9939-2012-11380-2},
       URL = {http://dx.doi.org/10.1090/S0002-9939-2012-11380-2},
}



\bib{G2}{article}{
    AUTHOR = {Gross, Leonard},
     TITLE = {A {P}oincar\'e lemma for connection forms},
   JOURNAL = {J. Funct. Anal.},
  FJOURNAL = {Journal of Functional Analysis},
    VOLUME = {63},
      YEAR = {1985},
    NUMBER = {1},
     PAGES = {1--46},
      ISSN = {0022-1236},
     CODEN = {JFUAAW},
   MRCLASS = {53C80 (53C05 58A10 81E20)},
  MRNUMBER = {795515 (87a:53110)},
MRREVIEWER = {Ng{\^o} Van Qu{\^e}},
       DOI = {10.1016/0022-1236(85)90096-5},
       URL = {http://dx.doi.org/10.1016/0022-1236(85)90096-5},
}


\bib{HT1}{article}{
    AUTHOR = {Hong, Min-Chun},
     AUTHOR = {Tian, Gang},
     TITLE = {Global existence of the {$m$}-equivariant {Y}ang-{M}ills flow
              in four dimensional spaces},
   JOURNAL = {Comm. Anal. Geom.},
  FJOURNAL = {Communications in Analysis and Geometry},
    VOLUME = {12},
      YEAR = {2004},
    NUMBER = {1-2},
     PAGES = {183--211},
      ISSN = {1019-8385},
   MRCLASS = {53C44 (53C07 58E15)},
  MRNUMBER = {MR2074876 (2005e:53103)},
MRREVIEWER = {J{\"u}rgen Eichhorn},
}

\bib{HT2}{article}{
    AUTHOR = {Hong, Min-Chun},
     AUTHOR = {Tian, Gang},
     TITLE = {Asymptotical behaviour of the {Y}ang-{M}ills flow and singular
              {Y}ang-{M}ills connections},
   JOURNAL = {Math. Ann.},
  FJOURNAL = {Mathematische Annalen},
    VOLUME = {330},
      YEAR = {2004},
    NUMBER = {3},
     PAGES = {441--472},
      ISSN = {0025-5831},
     CODEN = {MAANA},
   MRCLASS = {53C44 (53C07)},
  MRNUMBER = {MR2099188 (2006h:53063)},
}
\bib{HS}{article}{
    AUTHOR = {Hundertmark, Dirk},
    AUTHOR = {Simon, Barry},
     TITLE = {A diamagnetic inequality for semigroup differences},
   JOURNAL = {J. Reine Angew. Math.},
  FJOURNAL = {Journal f\"ur die Reine und Angewandte Mathematik},
    VOLUME = {571},
      YEAR = {2004},
     PAGES = {107--130},
      ISSN = {0075-4102},
     CODEN = {JRMAA8},
   MRCLASS = {47D08 (35J10 47F05 81Q10)},
  MRNUMBER = {MR2070145 (2005d:47078)},
MRREVIEWER = {Michael J. Gruber},
}

\bib{Jac}{book}{
    AUTHOR = {Jackson, John David},
     TITLE = {Classical electrodynamics},
   EDITION = {Third},
 PUBLISHER = {John Wiley \& Sons Inc.},
   ADDRESS = {New York},
      YEAR = {1999},
     PAGES = {xxii+808},
   MRCLASS = {78.00},
  MRNUMBER = {MR0436782 (55 \#9721)},
MRREVIEWER = {T. Kahan},
}



\bib{LM}{book}{
    AUTHOR = {Lions, J.-L.},
    AUTHOR = {Magenes, E.},
     TITLE = {Non-homogeneous boundary value problems and applications.
              {V}ol. {I}},
      NOTE = {Translated from the French by P. Kenneth,
              Die Grundlehren der mathematischen Wissenschaften, Band 181},
 PUBLISHER = {Springer-Verlag, New York-Heidelberg},
      YEAR = {1972},
     PAGES = {xvi+357},
   MRCLASS = {35JXX (35KXX 35LXX 46E35)},
  MRNUMBER = {0350177 (50 \#2670)},
}

\bib{L1}{article}{
    AUTHOR = {L{\"u}scher, Martin},
     TITLE = {Trivializing maps, the {W}ilson flow and the {HMC} algorithm},
   JOURNAL = {Comm. Math. Phys.},
  FJOURNAL = {Communications in Mathematical Physics},
    VOLUME = {293},
      YEAR = {2010},
    NUMBER = {3},
     PAGES = {899--919},
      ISSN = {0010-3616},
     CODEN = {CMPHAY},
   MRCLASS = {81T25 (81T13 81T17 81T80)},
  MRNUMBER = {2566166 (2011d:81217)},
MRREVIEWER = {Axel Maas},
       DOI = {10.1007/s00220-009-0953-7},
       URL = {http://dx.doi.org/10.1007/s00220-009-0953-7},
}

\bib{L2}{article}{
    AUTHOR = {L{\"u}scher, Martin},
     TITLE = {Properties and uses of the {W}ilson flow in lattice {QCD}},
   JOURNAL = {J. High Energy Phys.},
  FJOURNAL = {Journal of High Energy Physics},
      YEAR = {2010},
    NUMBER = {8},
     PAGES = {071, 18},
      ISSN = {1029-8479},
   MRCLASS = {81V05 (81T25)},
  MRNUMBER = {2756058 (2011m:81308)},
       DOI = {10.1007/JHEP08(2010)071},
       URL = {http://dx.doi.org/10.1007/JHEP08(2010)071},
}

\bib{L3}{article}{
    AUTHOR = {L{\"u}scher, Martin},
     TITLE = {Chiral symmetry and the {Y}ang-{M}ills gradient flow},
   JOURNAL = {J. High Energy Phys.},
  FJOURNAL = {Journal of High Energy Physics},
      YEAR = {2013},
    NUMBER = {4},
     PAGES = {123, front matter + 39},
      ISSN = {1126-6708},
   MRCLASS = {81T15 (81T25 81V05)},
  MRNUMBER = {3065842},
}

\bib{LW}{article}{
    AUTHOR = {L{\"u}scher, Martin},
    AUTHOR = {Weisz, Peter},
     TITLE = {Perturbative analysis of the gradient flow in non-abelian
              gauge theories},
   JOURNAL = {J. High Energy Phys.},
  FJOURNAL = {Journal of High Energy Physics},
      YEAR = {2011},
    NUMBER = {2},
     PAGES = {051, i, 22},
      ISSN = {1029-8479},
   MRCLASS = {81T25 (81T13 81V05)},
  MRNUMBER = {2820807},
       DOI = {10.1007/JHEP02(2011)051},
       URL = {http://dx.doi.org/10.1007/JHEP02(2011)051},
}

\bib{Pol1}{article}{
    AUTHOR = {Polyakov, A. M.},
     TITLE = {String representations and hidden symmetries for gauge fields},
   JOURNAL = {Phys. Lett.},
  FJOURNAL = {Physics Letters},
    VOLUME = {82B},
      YEAR = {1979},
    NUMBER = {2},
     PAGES = {247--250},
      ISSN = {0029-5582},
     CODEN = {NUPBBO},
   MRCLASS = {81G05 (81E99)},
}


\bib{Pol2}{article}{
    AUTHOR = {Polyakov, A. M.},
     TITLE = {Gauge fields as rings of glue},
   JOURNAL = {Nuclear Phys. B},
  FJOURNAL = {Nuclear Physics. B},
    VOLUME = {164},
      YEAR = {1980},
    NUMBER = {1},
     PAGES = {171--188},
      ISSN = {0029-5582},
     CODEN = {NUPBBO},
   MRCLASS = {81G05 (81E99)},
  MRNUMBER = {561638 (81c:81060)},
MRREVIEWER = {Jorge Andr{\'e} Swieca},
       DOI = {10.1016/0550-3213(80)90507-6},
       URL = {http://dx.doi.org/10.1016/0550-3213(80)90507-6},
}

\bib{Po}{book}{
    AUTHOR = {Poor, Walter A.},
     TITLE = {Differential geometric structures},
 PUBLISHER = {McGraw-Hill Book Co.},
   ADDRESS = {New York},
      YEAR = {1981},
     PAGES = {xiii+338},
      ISBN = {0-07-050435-0},
   MRCLASS = {53-01 (53-02 53C21)},
  MRNUMBER = {647949 (83k:53002)},
MRREVIEWER = {N. J. Hitchin},
}


\bib{Ra}{article}{
    AUTHOR = {R{\aa}de, Johan},
     TITLE = {On the {Y}ang-{M}ills heat equation in two and three
              dimensions},
   JOURNAL = {J. Reine Angew. Math.},
  FJOURNAL = {Journal f\"ur die Reine und Angewandte Mathematik},
    VOLUME = {431},
      YEAR = {1992},
     PAGES = {123--163},
      ISSN = {0075-4102},
     CODEN = {JRMAA8},
   MRCLASS = {58E15 (53C07 58G11)},
  MRNUMBER = {MR1179335 (94a:58041)},
MRREVIEWER = {Dennis M. DeTurck},
}

\bib{Sa}{article}{
    AUTHOR = {Sadun, Lorenzo Adlai},
     TITLE = {Continuum regularized Yang-Mills theory},
   JOURNAL = {Ph. D. Thesis, Univ. of California, Berkeley},
  FJOURNAL = {Communications in Mathematical Physics},
      YEAR = {1987},
      PAGES={67+ pages}
}


\bib{Sei}{book}{
    AUTHOR = {Seiler, Erhard},
     TITLE = {Gauge theories as a problem of constructive quantum field
              theory and statistical mechanics},
    SERIES = {Lecture Notes in Physics},
    VOLUME = {159},
 PUBLISHER = {Springer-Verlag},
   ADDRESS = {Berlin},
      YEAR = {1982},
     PAGES = {v+192},
      ISBN = {3-540-11559-5},
   MRCLASS = {81E08 (81-02 81E25 82A05)},
  MRNUMBER = {MR785937 (86g:81084)},
MRREVIEWER = {Claus Montonen},
}

\bib{Tay3}{book}{
    AUTHOR = {Taylor, Michael E.},
     TITLE = {Partial differential equations. {III}},
    SERIES = {Applied Mathematical Sciences},
    VOLUME = {117},
      NOTE = {Nonlinear equations,
              Corrected reprint of the 1996 original},
 PUBLISHER = {Springer-Verlag},
   ADDRESS = {New York},
      YEAR = {1997},
     PAGES = {xxii+608},
      ISBN = {0-387-94652-7},
   MRCLASS = {35-01 (46N20 47N20 58Gxx)},
  MRNUMBER = {MR1477408 (98k:35001)},
MRREVIEWER = {Luigi Rodino},
}

\bib{Z}{article}{
    AUTHOR = {Zwanziger, Daniel},
     TITLE = {Covariant quantization of gauge fields without {G}ribov
              ambiguity},
   JOURNAL = {Nuclear Phys. B},
  FJOURNAL = {Nuclear Physics. B},
    VOLUME = {192},
      YEAR = {1981},
    NUMBER = {1},
     PAGES = {259--269},
      ISSN = {0029-5582},
     CODEN = {NUPBBO},
   MRCLASS = {81E10 (53C05 58D30)},
  MRNUMBER = {MR635216 (82k:81062)},
}

\end{biblist}
\end{bibdiv}
\end{document}